\documentclass[conference]{IEEEtran}

\usepackage{color}
\usepackage{epsf}
\usepackage{times}
\usepackage{epsfig}
\usepackage{graphicx}
\usepackage{amsmath}
\usepackage{amssymb}
\usepackage{amsxtra}
\usepackage{here}
\usepackage{times}
\usepackage{url}
\usepackage{cite}
\usepackage[ruled,vlined]{algorithm2e}
\usepackage{subcaption}
\usepackage{epstopdf}
\usepackage{amsmath,amsthm}
\usepackage{dsfont}
\usepackage{verbatim} 
\usepackage{lettrine}
\usepackage{gensymb}

\newtheorem{theorem}{\bf Theorem}

\newtheorem{proposition}{\bf Proposition}
\newtheorem{lemma}{\bf Lemma}

\newtheorem{definition}{\bf Definition}

\makeatletter

\IEEEoverridecommandlockouts

\begin{document}
	
\title{\huge Drones in Distress: A Game-Theoretic Countermeasure for Protecting UAVs Against GPS Spoofing}
\author{\IEEEauthorblockN{AbdelRahman Eldosouky, Aidin Ferdowsi, and  Walid Saad}\vspace{-0.3cm}\\
	\IEEEauthorblockA{
		\small Wireless@VT, Bradley Department of Electrical and Computer Engineering, Virginia Tech, Blacksburg, VA, USA,\\Emails:\{iv727,aidin,walids\}@vt.edu\\
		\thanks{This research was supported by the US National Science Foundation under Grants OAC-1541105 and OAC-1638283.}
	}\vspace{-0.9cm}}

\maketitle\vspace{-0.8cm}
\begin{abstract}
One prominent security threat that targets unmanned aerial vehicles (UAVs) is the capture via GPS spoofing in which an attacker manipulates a UAV's global positioning system (GPS) signals in order to capture it. Given the anticipated widespread deployment of UAVs for various purposes, it is imperative to develop new security solutions against such attacks.  In this paper, a mathematical framework is introduced for analyzing and mitigating the effects of GPS spoofing attacks on UAVs. In particular, system dynamics are used to model the optimal routes that the UAVs will adopt to reach their destinations. The GPS spoofer's effect on each UAV's route is also captured by the model. To this end, the spoofer's optimal imposed locations on the UAVs, are analytically derived; allowing the UAVs to predict their traveling routes under attack. Then, a countermeasure mechanism is developed to mitigate the effect of the GPS spoofing attack. The countermeasure is built on the premise of cooperative localization, in which a UAV can determine its location using nearby UAVs instead of the possibly compromised GPS locations. To better utilize the proposed defense mechanism, a dynamic Stackelberg game is formulated to model the interactions between a GPS spoofer and a drone operator. In particular, the drone operator acts as the leader that determines its optimal strategy in light of the spoofer's expected response strategy. The equilibrium strategies of the game are then analytically characterized and studied through a novel proposed algorithm. Simulation results show that, when combined with the Stackelberg strategies, the proposed defense mechanism will outperform baseline strategy selection techniques in terms of reducing the possibility of UAV capture.
\end{abstract} 
\vspace{0.2cm}


\section{Introduction}\label{Sec:Intro}

Unmanned aerial vehicles (UAVs), popularly known as drones, have been recently adopted in many Internet of Things (IoT) systems to provide services such as telecommunications, delivery, surveillance, and medical services~\cite{mozaffari2016unmanned,rahmati2019dynamic, zhang2018downlink, hodgson2016precision,   eldosouky2017resilient,chen2017caching,mozaffari2017mobile}. Due to their ability to hover and their high-mobility capability without being restricted to specific routes, UAVs can provide services in hard-to-reach locations such as natural disaster sites. Considering their ease of deployment, UAVs can play a major rule in time-critical systems~\cite{french2018environment} and to provide urgent Internet and communication services when necessary~\cite{mozaffari2019beyond} and~\cite{amer2018caching}.

However, the widespread use of UAVs in different applications exposes them to a plethora of security threats that include cyber, physical, and cyber-physical attacks~\cite{altawy2017security}. Examples include cyber attacks such as, false data injection~\cite{challita2019machine} and physical attacks such as targeting the UAVs using firearms or hunting rifles. Cyber-physical attacks, on the other hand, represent a category of sophisticated attacks that aim at causing both cyber and physical damage to the UAV such as GPS spoofing, GPS jamming, hijacking the connection between a UAV and its controller, and thwarting delivery drones~\cite{sanjab2019game}.

Among these attacks, GPS spoofing is seen as one of the most imminent threats as it is practical and can be easily performed against UAVs~\cite{shepard2012evaluation}. In GPS spoofing attacks, an attacker transmits fake GPS signals to a UAV with slightly higher power than the authentic GPS signals, so as to mislead the UAV into thinking it is in another location. Hence, the attacker can use this technique to send the UAV to another, predetermined, location where it can be captured, thus executing a capture via GPS spoofing attack~\cite{kerns2014unmanned}. The authors in~\cite{kerns2014unmanned} discussed two types of GPS spoofing attacks known as covert and overt attacks. In a covert attack, the attacker wants to avoid triggering some spoofing detection techniques such as jamming-to-noise ratio and frequency unlock monitoring within the GPS receiver. This requires the attacker to be capable of accurately monitoring the target UAV and to transmit its spoofing signals with specific powers and frequencies. The attacker may also be forced to limit the changes it can impose on a UAV. In contrast, in an overt attack, the attacker can impose any location on the UAV with the risk of being detected.

\subsection{Related Works}

Different techniques have been proposed in literature to defend against GPS spoofing attacks with a focus on attack detection~\cite{warner2003gps,iqbal2008legal,o2012real,schmidt2016survey,jansen2018crowd}. In~\cite{warner2003gps}, different techniques are discussed that can enable a UAV's receiver to detect the spoofing attacks. This includes allowing the receiver to observe the received signal strength and compare it to the expected signal strength over time. It can also monitor the identification codes of GPS satellites or keep checking the time intervals to see if they are constant. While these techniques can help to detect basic attackers, they fail against sophisticated attacks in which the attacker monitors the target object accurately~\cite{iqbal2008legal}. In~\cite{o2012real}, the authors proposed a method to detect GPS spoofing attacks by using two GPS receivers and checking their cross-correlations. This method was tested against several spoofing attacks and was shown to successfully detect attacks, however, it has low probability of differentiating spoofed from authentic GPS signals and cannot detect the spoofing when the signals are weak.

Other techniques to thwart GPS spoofing such as receiver autonomous integrity monitoring, signal to interference ratio, Doppler shift detection are discussed in~\cite{schmidt2016survey}. However, all of these techniques can be avoided by highly capable adversaries that can carefully generate GPS counterfeit signals to avoid triggering these detection schemes. In~\cite{akos2012s}, automatic gain control is used within the GPS receiver to detect and flag potential spoofing attacks within a low computational complexity framework. Finally, the work in~\cite{jansen2018crowd} proposed a technique that allows UAVs to detect GPS spoofers by using an independent ground infrastructure that continuously analyzes the contents and times of arrival of the estimated UAV positions. The proposed technique was shown to be effective in detecting the spoofing attacks in less than two seconds and to determine the spoofer's location after $15$ minutes of monitoring time, with an accuracy of up $150$ meters.

Other works in literature have studied the use of multiple receivers to detect GPS spoofing attacks such as in~\cite{montgomery2011receiver,jansen2016multi,heng2015gps}. In~\cite{montgomery2011receiver}, the authors demonstrated the ability of using a dual antenna receiver in detecting GPS spoofing attacks. Their proposed technique depends on observing the carrier differences between the different antennas referenced to the same oscillator. Under the proposed configuration, an attacker will need to use an additional transmitting antenna for every additional receiver antenna which complicates the attacker's mission. In~\cite{jansen2016multi}, multiple independent GPS receivers were used to detect GPS spoofing attacks. The proposed technique depends on fixing the distances between the receivers and then measuring the distances between the receivers' reported locations. Under authentic GPS signals, the measured distances will be similar to the previously fixed distances. However, under a GPS spoofing attack, the measured distances will be very close to zero, as all the receivers are spoofed with the same fake location. Finally, in~\cite{heng2015gps}, multiple receivers were used to authenticate the GPS signals based on the correlation with the military GPS signal, without the need to decrypt it. Among these receivers, one GPS receiver uses the other receivers, referred to as cross-check receivers, to determine whether its GPS signals are authentic. The proposed technique was shown to be effective even when the cross-check receivers are spoofed with some probability. The technique was tested with stationary and moving GPS receivers and was shown to effectively detect the spoofing attacks.

However, one limitation of these existing GPS spoofing detection techniques, i.e.~\cite{warner2003gps,iqbal2008legal,o2012real,schmidt2016survey,akos2012s,jansen2018crowd,montgomery2011receiver,jansen2016multi,heng2015gps}, is that they do not provide an approach to determine the real location of the UAV, after detecting the attack. Thus, if a UAV is attacked while following a route towards a specific destination, the best it can do is to recognize the attack and to stop using the altered GPS signals. However, the UAV will not be able to determine its real location, and, hence, it will not resume its motion towards the specified destination. Indeed, these prior works are mostly focused on detection techniques and do not provide any attack mitigation or defense mechanisms (beyond discarding GPS signals altogether).

\subsection{Contributions}
The main contribution of this paper is, thus, a general framework for UAVs to mitigate the effect of capture attacks via GPS spoofing.
Unlike the prior works~\cite{warner2003gps,iqbal2008legal,o2012real,schmidt2016survey,akos2012s,jansen2018crowd}, our framework can both allow the UAVs to detect the GPS spoofing attacks and to determine their real locations. This will enable the UAVs to avoid being captured and to resume their previous routes and fulfill their missions.

In the proposed framework, we use system dynamics~\cite{ferdowsi2018robust} to model a UAV's motion between its origin and destination. For this model, we derive the optimal UAV's controller that allows  UAV to travel on the shortest path between any given two locations. This model also captures the effect of a GPS spoofer's on the UAV's traveling path. In particular, we analytically derive the optimal location, under covert attack, that a GPS spoofer can impose on a UAV to lead it to the attacker's desired destination where it can be physically captured. We, then, introduce a defense mechanism built on the technique of cooperative localization~\cite{qu2011cooperative}, which enables a UAV, traveling within a group of proximate UAVs, to determine its location using the real locations of neighboring UAVs and their relative distances. The mechanism also allows the identification of which UAV is being attacked.

Subsequently, we model the interactions between a GPS spoofer and a group of UAVs using game theory~\cite{eldosouky2016single}. In particular, we formulate a dynamic Stackelberg game in which the drone operator is the leader that selects its strategy first, and then, the spoofer responds by selecting its strategy. A strategy, here, represents a set of actions taken over all the time steps. We, then, propose an efficient technique to solve the formulated dynamic Stackelberg game. Using this technique, we analytically derive the Stackelberg strategies for the game. Finally, through simulations, we show that drone operator can effectively use the proposed defense mechanism to protect the UAVs from being captured and minimize the attacker's effect on the UAVs' optimal routes.

In summary, our contributions include
\begin{itemize}
	\item We propose a general framework, using realistic system dynamics, to model a UAV's traveling path between any two locations. This model takes into account the effect of a possible GPS spoofer on the UAV's traveling path.
	\item  We analytically derive the attacker's optimal imposed location, on a UAV. This imposed location ensures that attack remains covert while maximizing the attacker's benefit from imposing a different location on the UAV.
	\item We propose a new defense mechanism built on the technique of cooperative localization to help UAVs to determine their real location, under GPS spoofing attack, using neighbor UAVs' real locations and their relevant distances.
	\item We then formulate a dynamic Stackelberg game to model the interaction between the drone operator and the GPS spoofer. This game formulation allows the drone operator to effectively use the proposed defense mechanism.
	\item We introduce a novel computationally-efficient approach to solve the formulated game and we analytically derive the Stackelberg equilibrium strategies for the game.
	\item We show, through simulations, that the derived Stackelberg strategies outperform other strategy selection techniques by reducing the UAVs possibility of being captured. The defense mechanism is also shown to mitigate the effects of GPS spoofing attacks on the UAVs' deflections from their planned routes.
	\end{itemize}

The rest of this paper is organized as follows. The UAV's system dynamics model is presented in Section~\ref{Sec:Model}. The attacker's model and the optimal imposed locations are derived in Section~\ref{Sec:Attack}. The proposed defense mechanism  and the Stackelberg dynamic game with its equilibrium solutions are formulated in Section~\ref{defenseMech}. Numerical results are presented and analyzed in Section~\ref{Sec:results}. Finally, conclusions are drawn in Section~\ref{Sec:conclusion}.

\section{System Model}\label{Sec:Model}
Consider a set $\mathcal{N}$ of $N$ UAVs performing a common mission, e.g., a drone delivery system responsible for delivering goods within a certain geographic area. Each UAV is typically equipped with a GPS receiver, a means of wireless communication, and other application-specific sensors. As it was shown in~\cite{su2016stealthy}, a GPS spoofing attack cannot affect the altitude of UAVs and, thus, we use a two-dimensional (2D) coordination system to specify their locations. Let the location of UAV $i$ at time $t$ be $\boldsymbol{x}_i(t) = \left[x_i(t), y_i(t)\right]^T$, where $i \in \mathcal{N}$. Similarly, the source locations, $O_i$, will be given as $\boldsymbol{x}_{O_i} = \left[x_{O_i}, y_{O_i}\right]^T$, and each destination's location is $\boldsymbol{x}_{d_i} = \left[x_{d_i}, y_{d_i}\right]^T$. Destinations are assumed to be fixed and not time dependent. The goal of each UAV is to minimize the transportation cost and, hence, it chooses the shortest path from its source $O_i$ towards its destination $d_i$. 


In our model, we consider an adversary located along the traveling paths of the UAVs whose goal is to spoof the GPS signals of any of the UAVs in order to send it to another location where it can be captured. We consider a capable GPS spoofer that can spoof from a distance (in the order of hundreds of meters) without the need to be co-located with the UAV's GPS receiver~\cite{kerns2014unmanned}. 


Prior to developing the threat model, we first use system dynamics to model the UAV's motion between its source location, the origin, and destination. This model is needed to better understand the impact of the attack on the UAV's mobility. In order for each UAV to minimize its travel time, each UAV will follow the shortest path between its current location and its destination which essentially consists of the straight line connecting the two locations in 2D space. Let the location of UAV $i$ after a time duration $\Delta$ be $\boldsymbol{x}_i(t+\Delta)$. Let $\boldsymbol{v}_i(t) = \left[v_{x_i}(t), v_{y_i}(t)\right]^T$ be the UAV's velocity at the beginning of time step $t$. The UAV's velocity at the end of the time step can be represented in terms of the UAV's acceleration as follows:
\begin{equation}\label{eq:nextVelocity}
\boldsymbol{v}_i(t+\Delta) = \boldsymbol{v}_i(t) + \Delta \cdot \boldsymbol{g}_i(t),
\end{equation}
where $\boldsymbol{g}_i(t)=\left[g_{x_i}(t), g_{y_i}(t)\right]^T$ is the acceleration of UAV $i$ during the time step starting at $t$ with a duration $T$.

The next location can then be represented using both the velocity and acceleration as follows:
\begin{equation}\label{eq:nextLocationAcceleration}
\boldsymbol{x}_i(t+\Delta) = \boldsymbol{x}_i(t) + \Delta \cdot \boldsymbol{v}_i(t) + \frac{\Delta^2}{2} \cdot \boldsymbol{g}_i(t).
\end{equation}

Since the force needed to move the UAV between two locations is proportional to both the UAV's acceleration and weight, and the UAV's weight is constant, this force can then be related directly to the acceleration~\cite{ferdowsi2018cyber}. Let $\boldsymbol{u}_i(t)=\left[u_{x_i}(t),u_{y_i}(t)\right]^T$ be the force needed to move the UAV between any two locations. This force will be proportional to the distance difference between the current and the next locations, i.e., the UAV must accelerate more in the direction with a larger distance difference. Fig.~\ref{UAV_travel} shows the UAV's traveling model with the force components in each direction. In Fig.~\ref{UAV_travel}, the distance difference in the $x$ direction is more than the difference in the $y$ direction, and, hence, the force component in the $x$ direction will be greater than that of the $y$ direction.

\begin{figure}[t]
	\centerline{\includegraphics[width=7cm]{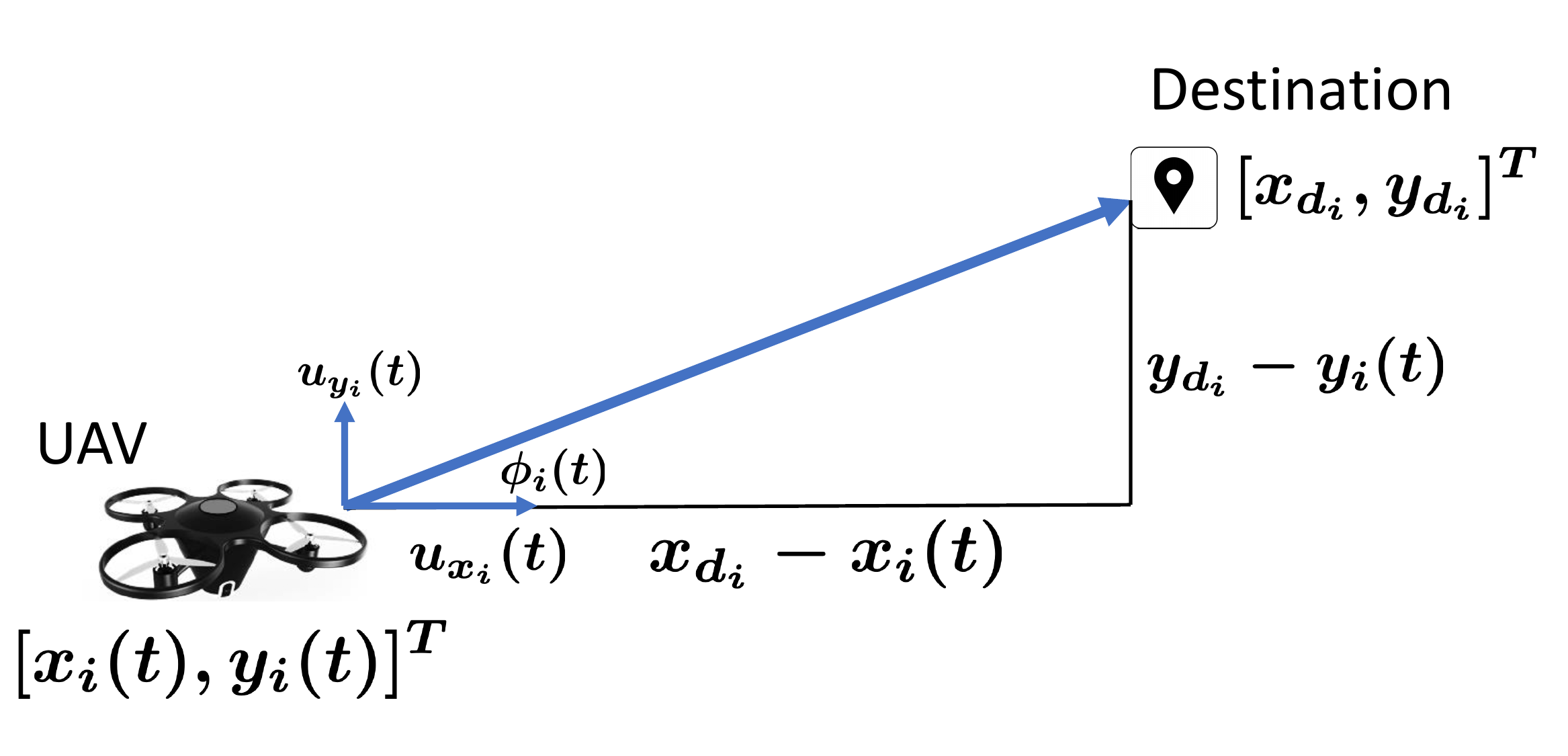}}
	\caption{UAV traveling model between two locations.}\label{UAV_travel}
	\vspace{-0.45cm}
\end{figure}


In order for each UAV to minimize its travel time, each UAV will need to find the optimal force to move between any two locations given that the maximum allowable force is $u^{\textrm{max}}$. Let $\phi_i$ be the angle between the UAV movement route and the positive $x$ direction which can be calculated as:
\begin{equation}\label{eq:phi}
\phi_i(t) =  \arctan\Big(\frac{y_{d_i} - y_i(t)}{x_{d_i} - x_i(t)}\Big)=\arctan(\gamma).
\end{equation}

The force components in both $x$ and $y$ directions can then be given by:
\begin{align}\label{eq:forcecomponents}
u_{x_i}^*(t) &= u^{\textrm{max}} \cdot \cos(\phi_i(t)), \nonumber \\
u_{y_i}^*(t) &= u^{\textrm{max}} \cdot \sin(\phi_i(t)) .
\end{align}

These values represent the optimal controller, i.e., the optimal force that each UAV can use to move between any two locations. Note that, if $\gamma >1$ then $\sin(\phi_i(t))>\cos(\phi_i(t))$ and in this case $u_{y_i}^*(t) > u_{x_i}^*(t)$, and vice verse. Substitute the optimal controller into (\ref{eq:nextLocationAcceleration}), the UAV's next location can then be given, in terms of the optimal controller, as:
\begin{align}\label{eq:optimalNextLocation}
x_i(t+\Delta) &= x_i(t) + \Delta \cdot v_{x_i}(t) + c \cdot \Delta^2 \cdot u_{x_i}^*(t), \nonumber\\
y_i(t+\Delta) &= y_i(t) + \Delta \cdot v_{y_i}(t) + c \cdot \Delta^2 \cdot u_{y_i}^*(t),
\end{align}
where $c = \frac{m}{2}$ is a constant and $m$ is the UAV's weight. Next we will discuss the effect of a GPS spoofer on the UAV's route by deriving both the optimal locations for an attacker to impose on a UAV and the manipulated routes under attack.

\section{UAV Traveling  Model under GPS Spoofing Attack}\label{Sec:Attack}


In our model, the GPS spoofer seeks to take control of the UAV's GPS antennas and then transmit tailored GPS signals to convince the UAV's navigation system that it is in a different location. The spoofer can perform either an overt attack or a covert attack.
In the overt attack, the spoofer makes no effort to hide its attack, it transmits its fake signals with higher power than the authentic GPS signals. The covert attack, on the other hand, requires an accurate tracking of the target UAV and the transmission of fake GPS signals with specific power requirements to avoid being immediately detected by the UAV.
Here, we consider a spoofer that wants to keep its attack covert by adjusting the transmission power of the counterfeit GPS signals to avoid being detected. Practical values for such power requirements can be found in~\cite{kerns2014unmanned}. In addition, the attacker will be limited to the changes it can impose on the UAV's location, each time, so that these imposed locations do not trigger the fault detectors within the UAV~\cite{su2016stealthy}.  The distance between the current and imposed location is known as the instance drifted distance~\cite{zeng2017practical}.

Let $e_{\textrm{max}}$ be the instance drifted distance that limits the spoofer's attack. Let $\boldsymbol{\hat{x}}_i(t) = \left[\hat{x}_i(t), \hat{y}_i(t)\right]^T$ be the attacker's imposed location on UAV $i$. Let $\boldsymbol{E}_i(t)= \left[e_{x_i}(t), e_{y_i}(t)\right]^T$ be a vector whose individual elements represent the distance difference
between the UAV's actual location and the attacker's imposed location. Then, we must have:
\begin{equation}\label{eq:normDistance}
\left\lVert \boldsymbol{E}_i(t) \right \rVert_2 = \left\lVert  \boldsymbol{x}_i(t) - \boldsymbol{\hat{x}}_i(t) \right \rVert_2 \le e_{\textrm{max}},
\end{equation}
which represents a circle of radius $e_{\textrm{max}}$ around the UAV's current location.





Note that, imposing an attacker-desired location on a UAV does not actually change the UAV's location, instead, it changes the UAV's belief about its location. This means that the UAV will still be in its real location but its navigation system will believe that it is in a different location. The UAV will then need to find a new optimal controller, i.e., new force components to move from its imposed location to its final destination. In this case, there will be two routes as shown in Fig.~\ref{fake_route}. Here, the upper route is the fake route which the UAV believes it is traveling on. This route starts from the attacker's imposed location towards the UAV's real destination. However, the UAV will actually travel on the lower path towards the attacker's desired destination. 

Let $\boldsymbol{x}^a_{d_i}=[x^a_{d_i}, y^a_{d_i}]^T$ be the attacker's desired destination for UAV $i$. The attacker's imposed location, $\boldsymbol{\hat{x}}_i(t)$, at each time step, needs to be calculated in order for the UAV to move towards $\boldsymbol{x}^a_{d_i}$. This can be achieved by satisfying the condition in the following lemma.

\begin{figure}[t]
	\centerline{\includegraphics[width=7cm]{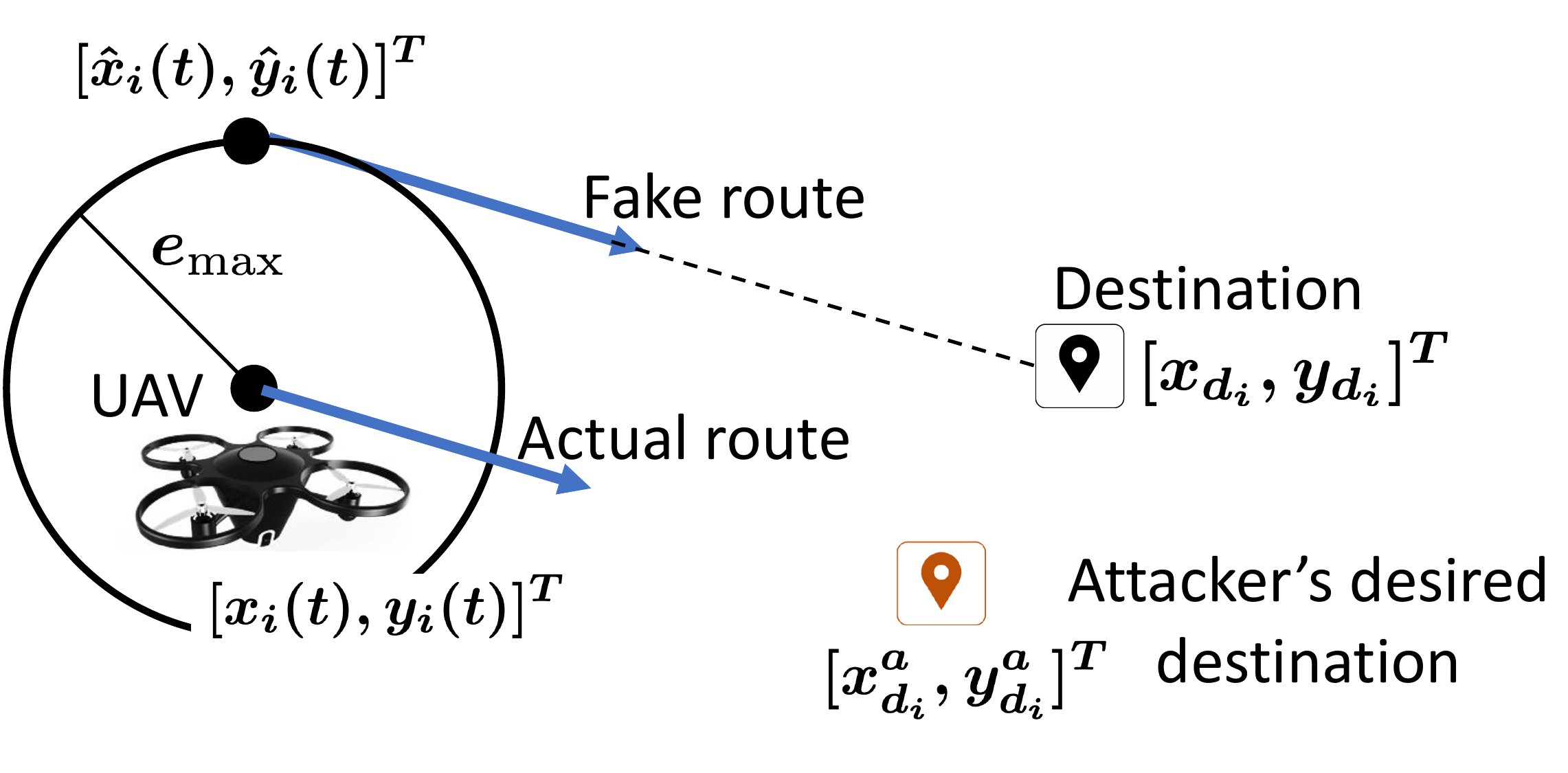}}
	\caption{UAV actual and fake routes.}\label{fake_route}
\end{figure}

\begin{lemma}\label{Lemma1}
	 The attacker's imposed location needs to satisfy $\hat{\gamma} = \gamma^a$, where $\hat{\gamma}=\Big( \frac{y_{d_i} - \hat{y}_i(t) }{x_{d_i} - \hat{x}_i(t)}\Big)$ and $\gamma^a = \Big( \frac{y^a_{d_i} - y_i(t) }{x^a_{d_i} - x_i(t)}\Big) $. 
\end{lemma}

\begin{proof} 
	Under an attack, the UAV will believe it is traveling from the attacker's imposed location, on the fake route in Fig.~\ref{fake_route}. The attacker should then select this imposed location $[\hat{x}_i(t),\hat{y}_i(t)]$ such that the actual route leads the UAV to the attacker's desired destination. Since, the UAV travels on the shortest path between any two locations, the actual route will represent a straight line that is parallel to the route which the UAV believes it is traveling on, i.e., the fake route.
	
	From Fig.~\ref{fake_route}, the fake route, can be defined by the two points $[\hat{x}_i(t),\hat{y}_i(t)]$ and $[x_{d_i},y_{d_i}]$. Similarly, the actual route, can be defined by the two points $[x_i(t),y_i(t)]$ and $[x^a_{d_i},y^a_{d_i}]$. For these routes to be parallel, the slopes of both routes need to be equal, i.e., $\hat{\gamma} = \gamma^a$.
\end{proof} 

Note that, under overt attack, according to Lemma~\ref{Lemma1}, the attacker can impose, theoretically, any location on the UAV that will lead the UAV to follow a path towards the attacker's desired destination. However, under a covert attack, the imposed location will be limited by (\ref{eq:normDistance}). This imposes constraints on the attacker when choosing the imposed location as there may be multiple or no points inside the circle, in \eqref{eq:normDistance}, that satisfy Lemma~\ref{Lemma1}.  When no such points exist, the best option for the attacker is to force the UAV to move in a direction as close as possible to the line connecting the real location and the attacker's desired destination, i.e., a direction that minimizes the difference $|\hat{\gamma} - \gamma^a|$ in Lemma~\ref{Lemma1}. If there are more than one point that satisfy Lemma~\ref{Lemma1}, then the best for the attacker is to choose the furthest point from the UAV's real destination as this gives more flexibility for the attacker in changing the imposed locations in the future time steps. Thus the attacker's optimal imposed location can be given by the solution of the following constrained-optimization problem.
\begin{align}\label{eq:imposedLocation}
& \mathop {\min}_{\boldsymbol{x}^a_i(t)}  |\hat{\gamma} - \gamma^a|  , \\ 
\text{s. t.}  \hspace{0.7cm}& {x}^a_i(t) =  \operatorname*{argmax}_{\boldsymbol{x}^a_i(t)}    \left\lVert  \boldsymbol{x}_{d_i} - \boldsymbol{x}^a_i(t) \right \rVert_2, \nonumber \\ 
& \left\lVert \boldsymbol{E}_i(t) \right \rVert_2 = \left\lVert  \boldsymbol{x}_i(t) - \boldsymbol{x}^a_i(t) \right \rVert_2 \le e_{\textrm{max}}.	\nonumber 
\end{align}	

In the following theorem, we analytically derive the attacker's imposed location, under covert attack.

\begin{theorem}\label{Theorem1}
	Let $ s_i(t)\triangleq (d_i(t)+a_i(t)+l_i)(d_i(t)+a_i(t)-l_i)(d_i(t)-a_i(t)+l_i)(a_i(t)+l_i-d_i(t)) $ where 
	\begin{align}
	d_i(t) &\triangleq \sqrt{(x_i(t)-x_{d_i})^2+(y_i(t)-y_{d_i})^2},\\
	a_i(t) &\triangleq \sqrt{(x_i(t)-x^a_{d_i})^2+(y_i(t)-y^a_{d_i})^2},\\
	l_i &\triangleq \sqrt{(x^a_{d_i}-x_{d_i})^2+(y^a_{d_i}-y_{d_i})^2}.
	\end{align}
	Then, the attacker's imposed location is the solution for the following set of equations:
	\begin{align}
	\hat{y}_i(t) - y_{d_i} = \frac{\hat{y}_i(t)-y_{d_i}}{\hat{x}_i(t)-x_{d_i}}(\hat{x}_i(t)-x_{d_i})\label{eq:LP}\\
	\left(x_i(t)-\hat{x}_i(t)\right)^2 + \left(y_i(t)-\hat{y}_i(t)\right)^2 = e^2_{\textrm{max}},\label{eq:Circle}
	\end{align}
	if $e_{\textrm{max}}> \frac{1}{2a_i(t)} \sqrt{s_i(t)} $, or the following set of equations: 
	\begin{align}
	\left(\hat{x}_i(t)-\hat{x}_{d_i}(t)\right)^2 + \left(\hat{x}_i(t)-\hat{y}_{d_i}(t)\right)^2 &= e^2_{\textrm{max}} + d^2(t)\label{eq:tanget}\\
	\left(x_i(t)-\hat{x}_i(t)\right)^2 + \left(y_i(t)-\hat{y}_i(t)\right)^2 &= e^2_{\textrm{max}}
	\end{align}
	if $e_{\textrm{max}}\leq \frac{1}{2a_i(t)} \sqrt{s_i(t)} $.
\end{theorem}

\begin{proof} 
	In Fig.~\ref{imposedlocation}, we use a geometrical representation for the problem to help clarify our proof. Let $L_a$ be the line connecting the UAV's real location to the attacker's desired destination. This line represents the attacker's ideal route for the UAV to travel on. Let $L_p$ be the line parallel to $L_a$ and passes through the UAV's real destination and $ \varepsilon $ be the distance between these two lines. There are then two cases for line $L_p$.
	
	\emph{Case 1:}
	if the line $L_p$ touches or intersects with the circle, formed by the constraint, then the point or the set of points of the intersection will represent a solution for the first objective function. In this case, the difference $|\gamma^a - \hat{\gamma}|$ will be $0$, which is the minimum possible value. This case will happen if $e_{\textrm{max}}\ge \varepsilon $. Thus, next, we find the value of  $ \varepsilon $ using the known values of $ d_i(t) $, $ a_i(t) $, and $ l_i $. To this end, we find the area of triangle $ ADX $, $ s_{ADX}(t) $, using two ways: 1) $ s_{ADX}(t) = \frac{\varepsilon \cdot a_i(t)}{2} $ and 2) $ s_{ADX}(t) = \sqrt{s_i(t)}/4 $, using Heron's formula ~\cite{coxeter1969introduction}. Hence, we will have $ \varepsilon = \frac{1}{2a_i(t)} \sqrt{s_i(t)} $. Therefore, if $ e_{\textrm{max}} \ge \frac{1}{2a_i(t)} \sqrt{s_i(t)} $, then the attacker's optimal choice is the intersection of $ L_p $ with circle $ C $ where $ L_p $ can be given by \eqref{eq:LP} and $ C $ can be given by \eqref{eq:Circle}.
	
	As this intersection may consist of more than one point, let $\mathcal{S}$ represents the solution set so far. The optimal solution for the problem in (\ref{eq:imposedLocation}) can then be found by solving the second optimization problem in the first constraint, i.e.,  ${x}^a_i(t) =  \operatorname*{argmax}_{\boldsymbol{x}^a_i(t)}    \left\lVert  \boldsymbol{x}_{d_i} - \boldsymbol{x}^a_i(t) \right \rVert_2$~\cite{sinha2018review}. If $\mathcal{S}$ has only one point, then this point will be the solution to the first constraint, which is a point on the circle perimeter. However, if $\mathcal{S}$ has multiple points, the solution will be the point on the circle's perimeter on the opposite side from the UAV's real destination. In either cases, this solution point will, then, be the attacker's imposed location.
		
	\emph{Case 2:}
	This case represents a more general case when $L_p$ does not intersect the circle formed by the constraint, i.e.,  $ e_{\textrm{max}} > \frac{1}{2a_i(t)} \sqrt{s_i(t)} $. 
	 The solution to the objective function, in this case, will not lie on $L_p$, instead it will lie on another line $L$ that passes through the UAV's real destination and intersects the circle at one point. This line $L$ should make the smallest angle $\alpha$ with the line $L_p$, and, hence, it minimizes the objective function. Thus, the optimal imposed location by the attacker, in this case, is the intersection of circle $ C $ and a circle $ C^{\prime} $ with a radius of $ \sqrt{e^2_{\textrm{max}}+d_i^2(t)} $ with its center at the actual desired destination of the UAV. Circle $C $ can be represented formally as in \eqref{eq:tanget}.  Note that the two circles will always intersect in two points, however, only one of them will minimize the objective function and, hence, this point will also be a solution for the maximization problem in the first constraint.
\end{proof} 

\begin{figure}[t]
		\centering
	\begin{subfigure}[b]{0.55\textwidth}
		\centering
	\includegraphics[clip,width=7cm]{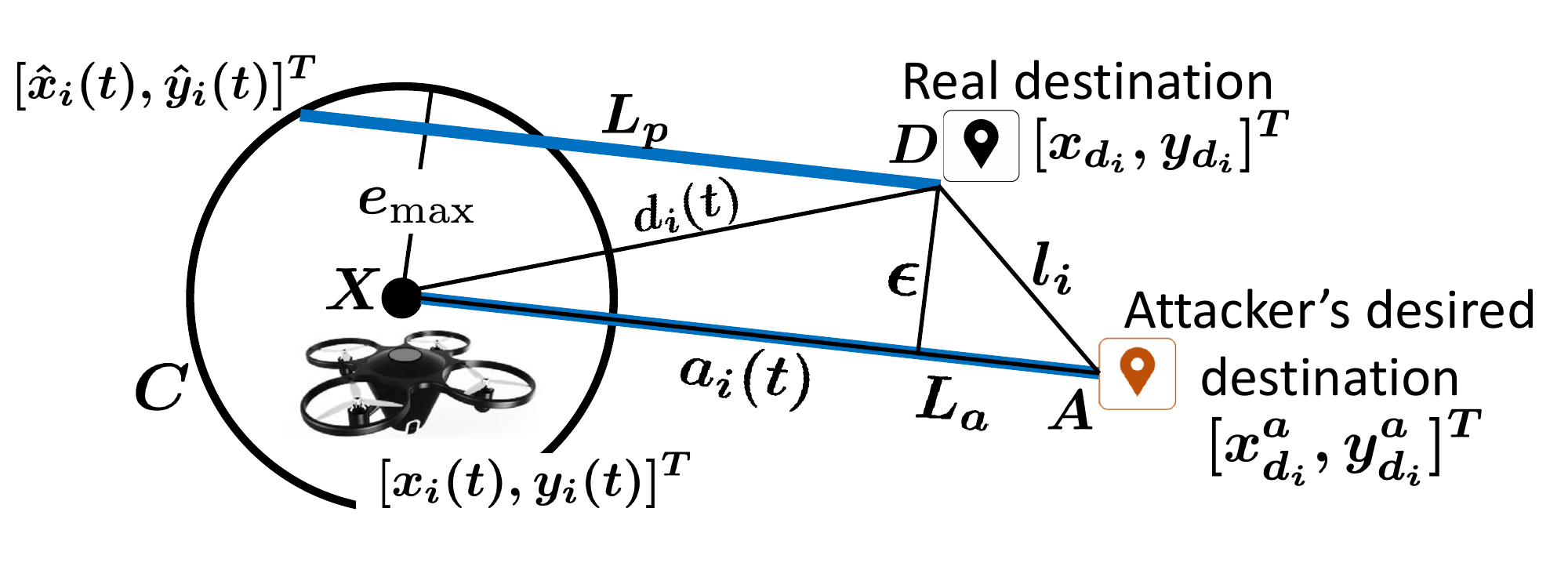}%
	\caption{Case 1}
	\end{subfigure}

	\begin{subfigure}[b]{0.55\textwidth}
				\centering
		\includegraphics[clip,width=7cm]{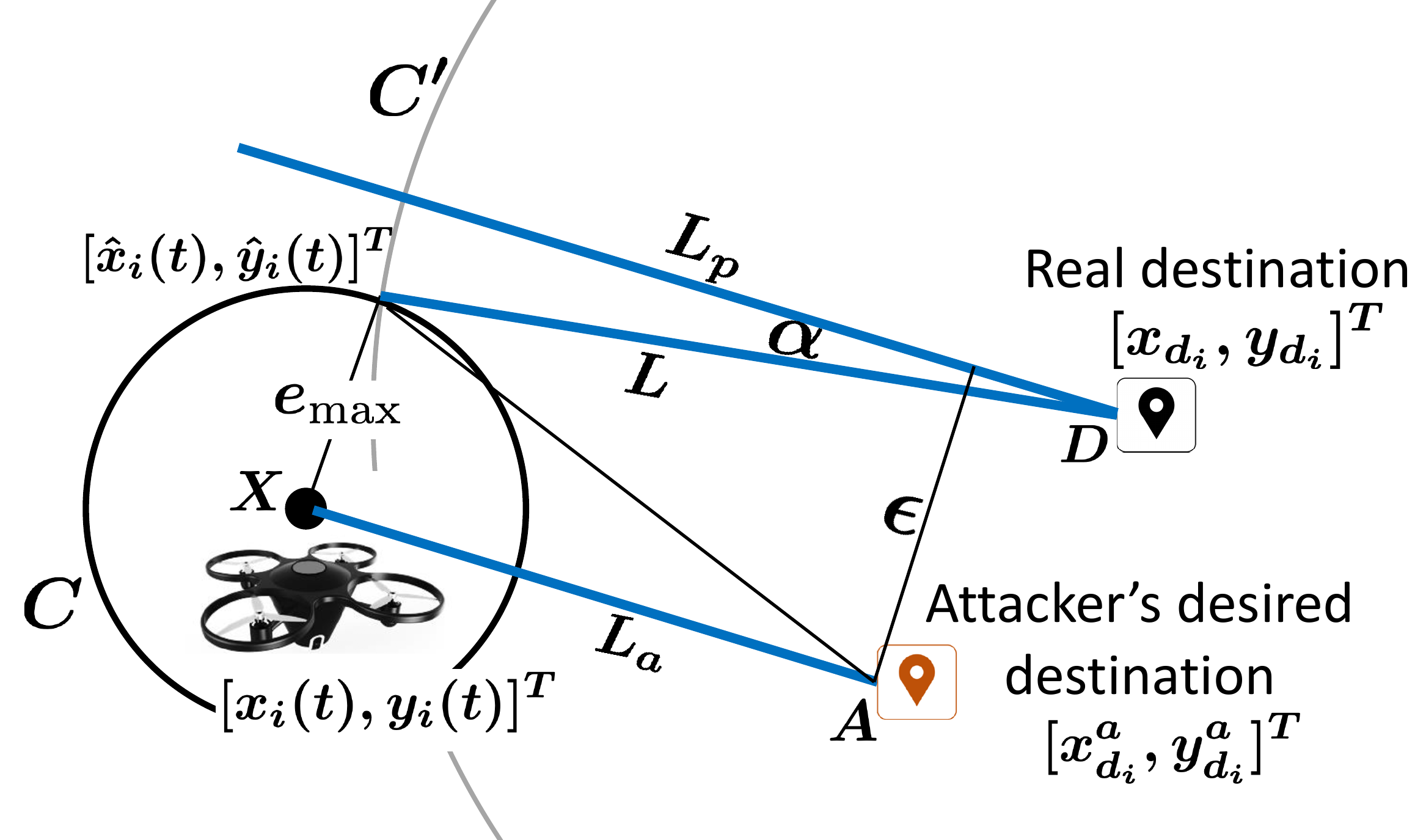}%
		\caption{Case 2}
	\end{subfigure}
	\caption{Determining the attacker's imposed location.}\label{imposedlocation}
\end{figure}


In the second case of Theorem \ref{Theorem1}, the attacker's imposed location will not lead the UAV directly to the attacker's desired destination. Consequently, the attacker might need to impose more than one location on the UAV along its perceived route. Each new imposed location can be calculated from Theorem \ref{Theorem1} with respect to the UAV's new location. Once the attacker can lead the UAV towards its desired destination, the imposed location, according to Theorem \ref{Theorem1}, will be the furthest point in the circle that maintains the same direction. Next, we study the UAV's manipulated route due to the attacker's imposed location.

Consider the UAV's next location under an attack. Similar to (\ref{eq:forcecomponents}), the UAV needs to compute the force components in both directions. The UAV thinks it is at the attacker's imposed location, $\boldsymbol{\hat{x}}_i(t)$, so it calculates the required force to move from $\boldsymbol{\hat{x}}_i(t)$ to its desired destination $\boldsymbol{x}_{d_i}$. Let $\phi^a_i$ be the angle between the $x$ direction and the line connecting the UAV's imposed location to its real destination. The value of $\phi^a_i$ can then be given as:
\begin{equation}\label{eq:phihat}
\phi^a(t) =  \arctan\Big(\frac{y_{d_i} -\hat{y}_i(t)}{x_{d_i} - \hat{x}_i(t)}\Big)=\arctan(\hat{\gamma}).
\end{equation}

Therefore, the force components in both $x$ and $y$ directions can then be given as:
\begin{align}\label{eq:forcecomponentshat}
u_{x_i}^a(t) &= u^{\textrm{max}} \cdot \cos(\phi^a_i(t)), \nonumber \\
u_{y_i}^a(t) &= u^{\textrm{max}} \cdot \sin(\phi^a_i(t)) .
\end{align}

The UAV will then use the optimal controller in \eqref{eq:forcecomponentshat} to move towards its real destination. However, the UAV will actually move from its real location not its perceived location as shown in Fig.~\ref{fake_route}.  Let $\boldsymbol{x}^a_i(t+1) = [x^a_i(t+1), y^a_i(t+1)]^T$ be the UAV's next location under attack. It can then be given, according to (\ref{eq:optimalNextLocation}), as:
\begin{align}\label{eq:nextLocationAttack}
x^a_i(t+\Delta) &= x_i(t) + \Delta \cdot v_{x_i}(t) + c \cdot \Delta^2 \cdot u_{x_i}^a(t), \nonumber\\
y^a_i(t+\Delta) &= y_i(t) + \Delta \cdot v_{y_i}(t) + c \cdot \Delta^2 \cdot u_{y_i}^a(t).
\end{align}

Note that, following the route calculated by (\ref{eq:nextLocationAttack}) may not guarantee the attacker to eventually lead the UAV to the attacker's desired destination. Achieving this depends on multiple parameters such as the UAV's current location and the locations of its real destination and the attacker's desired destination. In general, the attacker should choose its desired destination to satisfy the following condition.

\begin{proposition}\label{Proposition1}
	Under a covert attack, the attacker's desired destination should be located on the same side as the UAV's real destination, in terms of the direction with the largest difference between the UAV's current location and its real destination.
\end{proposition}

\begin{proof}
The proof is presented in Appendix \ref{Prop1proof}.
\end{proof}

Next, we will discuss the defense mechanism against the considered GPS spoofing attack.

\section{GPS Spoofing Countermeasure}\label{defenseMech}

\subsection{Defense Mechanism for Mitigating Spoofing Attacks}

We propose a defense mechanism built on the concept of cooperative localization~\cite{qu2011cooperative}  which is a framework that enables a UAV to determine its real location in a 2D coordinate system using the locations of three other UAVs. Each UAV is assumed to have a means of measuring its relative distances to the other, neighboring UAVs by inter-UAV range measurements. In cooperative localization, a UAV chooses any three neighboring UAVs, to update its location, given that the selected UAVs are non-collinear. Following this, the UAV can accurately determine its 2D location. While the cooperative localization mechanism in~\cite{qu2011cooperative} can help a UAV to determine its location, it was proposed to be used in case of GPS signals loss and cannot be used, directly, in case of GPS spoofing attack due to the different nature of the problem.

Under a covert GPS spoofing attack, a UAV cannot trust its GPS location nor the locations of other UAVs. Choosing a neighboring UAV for the cooperative localization mechanism will involve a risk as this UAV might itself be under attack. To overcome this limitation, we propose a defense mechanism based on the fact that a GPS spoofing attacker can target only one UAV at a time, as discussed earlier. In our proposed mechanism, a UAV will use the locations of four neighboring UAVs, instead of three, to determine its real location by identifying the UAV under attack and eliminating it from the calculations. The proposed mechanism has the same requirements of cooperative localization, i.e., the UAVs are non-collinear, a UAV can request other UAVs' locations through inter-UAV communications, and each UAV needs to be able to measure its relative distances to its neighboring UAVs.

Due to the fact that determining a 2D location requires only three UAVs, the fourth UAV will be used to check the results as follows. A UAV will calculate its location using all the permutations of three UAVs formed from the selected four UAVs. Let any UAV and its selected four neighbor UAVs represent a set given by $\mathcal{F}=\{F_i\}$, where $F_i$ represent a UAV and $i \in \{1, \dots, 5\}$.  Assume UAV $F_1$ wants to calculate its location, let its location calculated from the GPS signals be $ \tilde{\boldsymbol{x}}_1$. The UAV, $F_1$, cannot determine at this point if this location is real or a spoofed location. The UAV will then calculate its location four more times using all the groups formed of three UAVs out of the selected four UAVS. For example, $\tilde{\boldsymbol{x}}_2$ can be calculated using UAVs $F_2,F_3,$ and $F_4$, $\tilde{\boldsymbol{x}}_3$ can be calculated using UAVs $F_2,F_3,$ and $F_5$, and so on for $\tilde{\boldsymbol{x}}_4$ and $\tilde{\boldsymbol{x}}_5$.

The UAV can then determine its real location according to the following cases:
\begin{itemize}
	\item If there is no attack, the value of $\tilde{\boldsymbol{x}}_1$ will equal all the other values, i.e., $\tilde{\boldsymbol{x}}_i , i=2,\dots,5$, will all be the same.
	\item If UAV $F_1$ is under attack, then all the values $\tilde{\boldsymbol{x}}_i ,i=2,\dots,5$, will be equal but their value will not equal $\tilde{\boldsymbol{x}}_1$. In this case, the real location of UAV $F_1$ is the value calculated from its neighboring UAVs. 
	\item If another UAV, rather than UAV $F_1$, is under attack, then the value of $\tilde{\boldsymbol{x}}_1$ will equal only one of the four other values. The other three values will be the same and the UAV that contributed to calculating these values will be the one under attack. 
\end{itemize}

\begin{figure}[t]
	\centerline{\includegraphics[width=7.3cm]{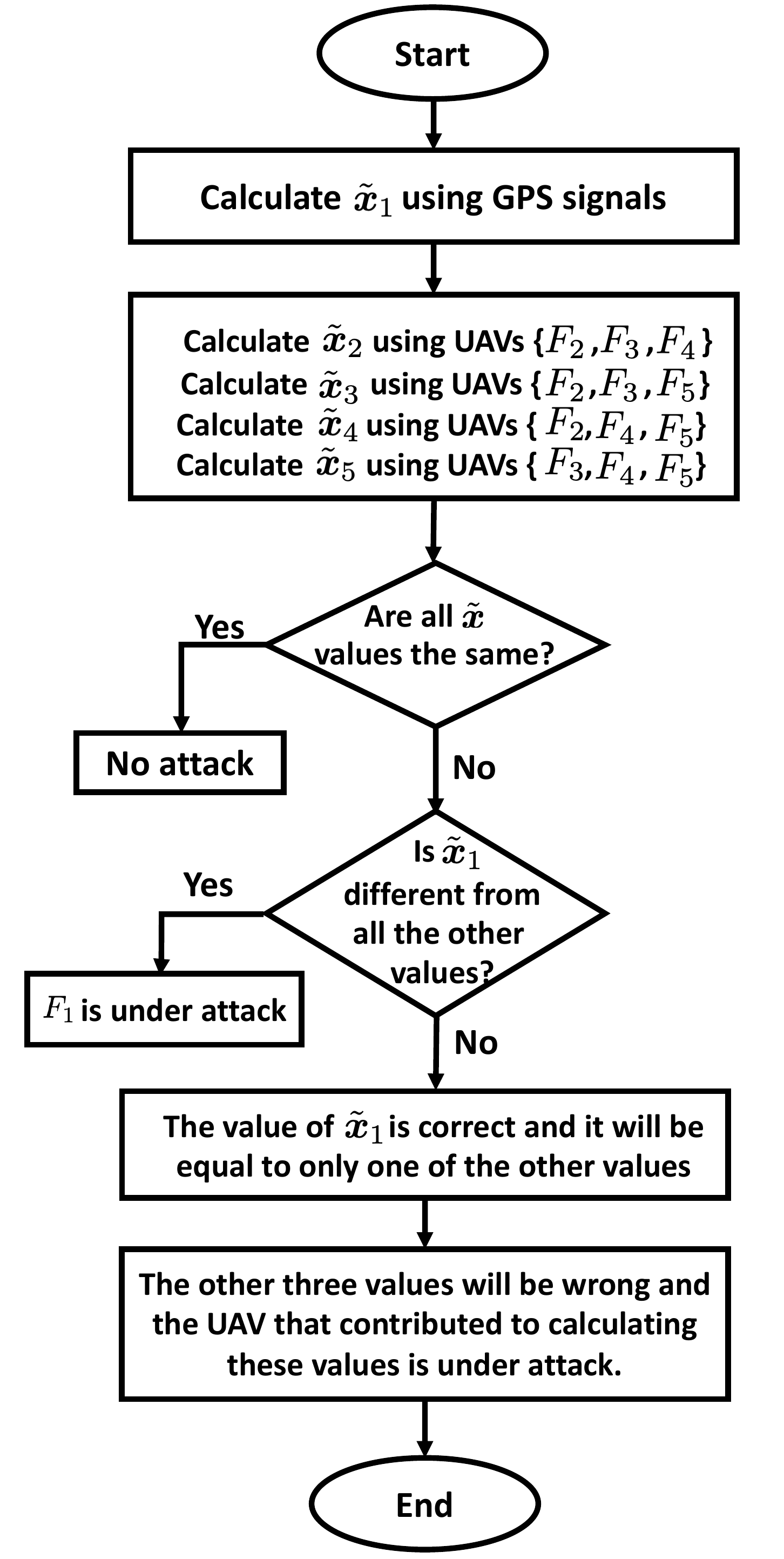}}
	\caption{Flowchart for the defense mechanism.}\label{fig:flowchart}

\end{figure}

Note that the technique used in~\cite{heng2015gps} was shown to require four different cross-check receivers, to detect the GPS spoofing attack, when $15\%-25\%$ of the cross-check receivers are unreliable. Comparing these findings to our proposed defense mechanism, the same number of UAVs, i.e., four cross-check receivers, will be required to detect a single attack, i.e., to detect that $25\%$ of the cross-check receivers are unreliable. However, our defense mechanism is able to not only detect attacks but to also determine the real locations.

Note that the proposed defense mechanism can similarly be extended to the case in which more than one UAV is being attacked simultaneously. The total number of UAVs and the complexity of determining the real locations are discussed in the following proposition.
\begin{proposition}\label{Prop2}
When $n$ UAVs are being attacked simultaneously, a total of $n+4$ UAVs, traveling in proximity, are needed such that each UAV can determine its real location. The complexity of finding the real location will be $\mathcal{O}{\binom{n+3}{3}}$ per UAV.
\end{proposition}

\begin{proof}
The case $n=1$ has been discussed in details earlier where four neighboring UAVs are required. Similarly, for the case in which $n >1$ UAVs are being attacked, four non-attacked neighboring UAVs are still enough to determine the real locations. In this case, if a UAV wants to determine its location, it will start by calculating its location using its GPS receivers and compare this location to the locations calculated from all possible groups of three neighboring UAVs. 

Since $n$ UAVs are under attack, there will be $4$ UAVs reporting their real locations. If these $4$ UAVs include the checking UAV, then its location, calculated from its GPS receivers, will match the location calculated from the group of the other $3$ UAVs. This will represent the real location of the checking UAV. On the other hand, if the $4$ UAVs, that are not attacked, do not include the checking UAV, then the locations calculated from the groups formed of these $4$ UAVs, i.e., $\binom{4}{3} = 4$ groups will give the same (real) location which will differ from the checking UAV's own location. This procedure will allow each checking UAV to determine whether it is under attack, to determine its real location, and to identify which UAVs are under attack.

The complexity of finding the real location depends on the total number of the locations to be compared which corresponds to the number of the groups formed from each $3$ UAVs, i.e., $\binom{(n+4)-1}{3} = \binom{n+3}{3}$, plus the UAV's own location. Therefore, the total number of calculations at each UAV will be $\binom{n+3}{3}+1$, which results in a complexity of $\mathcal{O}{\binom{n+3}{3}}$.
\end{proof}

In Proposition \ref{Prop2}, when the number of the attacked UAVs is very large, the time complexity might affect the applicability of the defense mechanism depending on the application. For instance, in time-sensitive applications, where the UAVs need to make rapid decisions, the increased number of calculations may hinder the UAVs' ability to determine their locations in a timely manner. On the other hand, if the goal is to lead the UAVs to their destinations without being constrained by time-sensitive tasks, then, the UAVs can still apply the defense mechanism even if the number of the attacked UAVs is large.

Finally, we summarize the basic case of our defense mechanism steps in the flow chart shown in Fig. ~\ref{fig:flowchart}. Note that, following our approach in Fig. \ref{fig:flowchart}, a UAV can determine its real location and identify which UAV is under attack. However, the UAV under attack will have to also execute the same procedure to determine its real location. One approach is to allow all the UAVs to continuously use the proposed defense mechanism along their travel paths. However, this might be challenging to do in long routes as the energy consumption due to exchanging communication messages, measuring distances to other UAVs, and calculating the locations will be significant compared to the UAVs' limited power.

Given that a GPS spoofer can target only one UAV at a time, we next propose a new approach to regulate the use of our proposed defense mechanism among the UAVs through studying the interactions between the GPS spoofer and a group of UAVs, managed by their operator. The drone operator wants to regulate the use of the defense mechanism, in an energy-efficient manner, by determining when each UAV can use it along its travel path, to avoid being captured. On the other hand, the GPS spoofer wants to take control of the UAVs to send them to other destinations where they can be captured. In doing so, both the operator and the spoofer will be affected by each others' actions. Therefore, we propose to use game theory~\cite{han2012game}  to model these interactions. In our game, the drone operator is the defender and the GPS spoofer is the attacker. As the attack and defense mechanisms are applied along each UAV's travel paths, the game will be time-depended, and, hence, a dynamic game model is appropriate.


\subsection{Dynamic Stackelberg Game Formulation}

We formulate a dynamic game in which every player, i.e., the attacker and the defender chooses its actions at every time step. Since the spoofer needs to monitor the targeted UAV to generate tailored spoofing signals, the spoofer will be able to attack only one UAV at any given time step. At each time, the attacker can choose one action out of set $\mathcal{Z}^a$ which represents the choice of one UAV to attack.
Similarly, the drone operator can choose one UAV to use the defense mechanism defined in Section~\ref{defenseMech} to update its location, given that the spoofer can attack only one UAV at a time. Let $\mathcal{Z}^d$ be the set from which the defender is choosing its actions, i.e., a UAV to apply the defense mechanism.


Note that, in the basic case of the proposed defense mechanism, each UAV needs the locations of four neighboring UAVs to determine its location. Therefore, each five UAVs can be seen as a separate group in which one UAV is applying the defense mechanism using the locations of the other four UAVs. For other cases in which more than one UAV is being attacked, the number of UAVs per group will be different. However, these cases are beyond the scope of this work as each case requires a separate analysis. In the following, we consider only the basic case of the defense mechanism in which each five, closely traveling, UAVs are considered to form a group that is applying the defense mechanism locally within the group. 
Hence, hereinafter, we consider a game in which the drone operator is protecting only one group, because the solution can be easily extended to the case of multiple groups of UAVs. 

The actions of each player, when taken at a time step, will affect the next locations of the UAVs. If a UAV is applying the defense mechanism, then it can accurately determine its location whether there is an attack or not. On the other hand, if a UAV is dependent on the GPS signals, it will be affected by the attacker's actions and its next location will depend on whether it is attacked or not. Here, we assume that each spoofing attack is successful in that the attacker will gain control of the UAV's GPS receivers and impose its desired location on the UAV's GPS. Let $z^a_i(t), i = 1,\dots,5$ be a variable indicating whether the attacker has chosen to attack UAV $i$, where $z^a_i(t)=1$ means the UAV $i$ is being attacked at time step $t$, and $z^a_i(t)=0$ otherwise. Similarly, let  $z^d_i(t), i = 1,\dots,5$, be a variable indicating whether the defender has chosen to protect UAV $i$, where $z^d_i(t)=1$ means the UAV $i$ is applying, at time step $t$, the defense mechanism, i.e., being protected, and $z^d_i(t)=0$ otherwise. Each UAV's next location can then be given by:
\begin{align}\label{NextLocationActions}
\boldsymbol{x}_i(t+\Delta,z^d_i(t),z^a_i(t)) = & z^d_i(t) \cdot \boldsymbol{x}_i(t) +  \big(1-z^d_i(t)\big) \big[z^a_i(t) \cdot \nonumber\\
& \boldsymbol{x}^a_i(t) + (1-z^a_i(t)) \cdot \boldsymbol{x}_i(t)\big],
\end{align}
where $\boldsymbol{x}_i(t)$ and $\boldsymbol{x}^a_i(t)$ are given by (\ref{eq:optimalNextLocation}) and (\ref{eq:nextLocationAttack}), respectively. Equation (\ref{NextLocationActions}) can be rearranged as:
\begin{align}\label{NextLocationActionsSimplified}
\boldsymbol{x}_i(t+\Delta,z^d_i(t),z^a_i(t)) = & \Big( 1- z^a_i(t) + z^a_i(t) \cdot z^d_i(t) \Big) \cdot \boldsymbol{x}_i(t)  \nonumber \\
 &+  \Big(z^a_i(t) -z^a_i(t) \cdot z^d_i(t)\Big)   \cdot \boldsymbol{x}^a_i(t).
\end{align}

From Theorem~\ref{Theorem1}, we can observe that the attacker's imposed location $\boldsymbol{x}^a_i(t)$  can be accurately calculated from the UAV's current location $\boldsymbol{x}_i(t)$, the UAV's real destination, and the attacker's desired destination. However, as the UAV's real destination and the attacker's desired destination are constants, the attacker's imposed location, at any given time step, can be given as a function of the UAV's real location. Therefore, the location in (\ref{NextLocationActionsSimplified}) can be given as a function in the UAV's current location and both player's actions, i.e., $\boldsymbol{x}_i(t+\Delta,z^d_i(t),z^a_i(t)) = f\big(\boldsymbol{x}_i(t),z^d_i(t),z^a_i(t)\big)$.

Next, we define the outcomes (utilities) for both players due to their interactions. Since the objective for each player is to move each UAV to its own desired destination, each player will take actions to minimize the distance between the current UAV's location and the player's desired destination. Thus, we define the utility function for the attacker, at each time step, as follows:
\begin{equation}\label{AttackersUtility}
U^a(t,z^d_i(t),z^a_i(t))= \sum_{i=1}^{5}  \left\lVert  \boldsymbol{x}^a_{d_i} - \boldsymbol{x}_i(t+\Delta,z^d_i(t),z^a_i(t)) \right \rVert^2_2.
\end{equation}

Similarly, the defender's utility, at each time step, can be given by:
\begin{align}\label{DefendersUtility}
U^d(t,z^d_i(t),z^a_i(t))= \sum_{i=1}^{5}  \left\lVert  \boldsymbol{x}_{d_i} - \boldsymbol{x}_i(t+\Delta,z^d_i(t),z^a_i(t)) \right \rVert^2_2.
\end{align}

Now, consider the players' actions and utilities over all time steps. Assume the maximum possible number of time steps is $\tau$, which is determined by the maximum time that any UAV can travel based on its fuel or battery. This number is known to the defender but the attacker does not need to know this number. From Proposition \ref{Proposition1}, the GPS spoofer will not be able to change the UAV's direction and, thus, once a UAV passes beyond the attacker's desired destination, the attacker will no more consider it when choosing its actions. Therefore, the game is considered to end for the attacker when all the UAVs pass beyond the attacker's desired destinations.

Consider the players' strategies which are defined as the players' actions taken at each time step $t$. Let $\boldsymbol{\beta}^a$ be an attacker's strategy defined by $\boldsymbol{\beta}^a =  \{z^a_i(1), \dots, z^a_i(\tau)\}$, and let $\mathcal{A}$ be the set of all the attacker's possible strategies. Similarly, let $\boldsymbol{\beta}^d$ be a defender's strategy defined by $\boldsymbol{\beta}^d = \{z^d_i(1), \dots, z^a_i(\tau)\}$, and let $\mathcal{D}$ be the set of all the defender's possible strategies. The attacker's accumulated utility will then be:
\begin{align}\label{AttackersAccUtility}
J^a(\boldsymbol{\beta}^d,\boldsymbol{\beta}^a)&= \sum_{t=1}^{\tau}  U^a(t\Delta,z^d_i(t),z^a_i(t)) \nonumber \\
&= \sum_{t=1}^{\tau}  \sum_{i=1}^{5}  \left\lVert  \boldsymbol{x}^a_{d_i} - \boldsymbol{x}_i(t\Delta,z^d_i(t),z^a_i(t)) \right \rVert^2_2.
\end{align}

Similarly, the defender's accumulated utility will be given by:
\begin{align}\label{DefendersAccUtility}
J^d(\boldsymbol{\beta}^d,\boldsymbol{\beta}^a)&= \sum_{t=1}^{\tau}  U^d(t\Delta,z^d_i(t),z^a_i(t)) \nonumber \\
&= \sum_{t=1}^{\tau}  \sum_{i=1}^{5}  \left\lVert  \boldsymbol{x}_{d_i} - \boldsymbol{x}_i(t\Delta,z^d_i(t),z^a_i(t)) \right \rVert^2_2.
\end{align}

To solve this dynamic game, we propose to use the  dynamic Stackelberg game model~\cite{simaan1973additional}. In Stackelberg games, one player, the leader, acts first by selecting its strategy and, then, the other player, the follower, can respond by selecting its strategy. In our game formulation, the drone operator will act as the leader as it can choose which UAVs to protect in advance and the attacker can observe this selection and responds by choosing which UAVs to attack.

Now, we can formally formulate a dynamic Stackelberg game $\mathit{\Xi}$ described by the tuple $\left\langle \mathcal{M}, \mathcal{A},  \mathcal{D},J^a, J^d, \tau \right\rangle$ where $\mathcal{M}$ is the set of the two players: the defender and the attacker, and the rest of the parameters as defined earlier. Based on the utility functions, the game is non-zero sum. This means, every player will try to minimize its utility and the sum of the utilities will not equal zero. Moreover, each player seeks to follow a strategy that minimizes its utility function given the other player's strategy. Next, we study our approach of finding the optimal strategies, for each player, under the formulated game.

\subsection{Stackelberg Game Solution }

The most commonly adopted solution for Stackelberg dynamic games is known as the Stackelberg equilibrium strategy concept~\cite{simaan1973additional}. This solution is given by a pair of strategies $(\boldsymbol{\beta}^{a*},\boldsymbol{\beta}^{d*})$  defined as follows.

\begin{definition}
The \emph{Stackelberg equilibrium strategies}, when the defender is the leader, are derived as follow. Let $r:\mathcal{D} \rightarrow \mathcal{A}$ be a mapping between the defender's strategies and the attacker's strategies, such that:
\begin{equation}\label{eq:attackerResponse}
J^a(\boldsymbol{\beta}^d,r(\boldsymbol{\beta}^d)) \le J^a(\boldsymbol{\beta}^d,\boldsymbol{\beta}^a) , \forall \boldsymbol{\beta}^a \in \mathcal{A},
\end{equation}
and the set:
\begin{equation}\label{eq:reactionset}
R^a = \{(\boldsymbol{\beta}^d,\boldsymbol{\beta}^a) \in  \mathcal{D} \times \mathcal{A}: \boldsymbol{\beta}^a = r(\boldsymbol{\beta}^d), \forall \boldsymbol{\beta}^d \in  \mathcal{D}\},
\end{equation}
is the reaction set for the attacker when the defender is the leader. The \emph{Stackelberg equilibrium strategies}  $(\boldsymbol{\beta}^{a*},\boldsymbol{\beta}^{d*}) \in R^a$ of the game should then satisfy:
\begin{equation}\label{eq:StackelbergSolution}
J^d(\boldsymbol{\beta}^{a*},\boldsymbol{\beta}^{d*}) \le J^d(\boldsymbol{\beta}^d,\boldsymbol{\beta}^a) , \forall (\boldsymbol{\beta}^d,\boldsymbol{\beta}^a) \in R^a.
\end{equation}
\end{definition}

Note that, solving for the Stackelberg equilibrium strategies that satisfy (\ref{eq:StackelbergSolution}) depends on the information available for each player, at each time step~\cite{basar1999dynamic}. According to~\cite{basar1999dynamic}, dynamic games can be solved using open-loop strategies, closed-loop strategies, or feedback strategies. In the formulated game, each player selects a strategy that minimizes its utility which involves taking actions, at each time step to control the UAVs' locations. In doing so, both players can observe the initial locations of the UAVs as well as their subsequent locations up to the current time step. This type of information coincides with the notion of \emph{closed-loop perfect information}~\cite{basar1999dynamic}, and, thus, we use closed-loop Stackelberg strategies to solve the formulated game. Note that the equilibrium strategies should satisfy \eqref{eq:StackelbergSolution} irrespective of the type of the solution.

Note that, the proposed Stackelberg solution is based on the assumption of complete information, i.e., both players have full information about their opponent's. This assumption should hold true for a powerful attacker that can accurately observe the drones. For the defender, it is assumed that the defender knows the attacker's desired destinations and, hence, can determine the attacker's reaction set in advance according to its role as a leader in the Stackelberg leader-follower scenario. To gain this information, one approach for the defender is to perform reconnaissance before launching the UAVs. Another practical solution is to observe the attacker's actions and to use machine learning, e.g., reinforcement learning to update its information about the UAVs' deviated routes under attack, and, hence, the directions of the attacker's desired destinations. However, this will require separate analysis and, hence, is left for future work. Here, we handle the incomplete information case for the defender when it cannot observe the attacker's actions. In this case, the defender might not be able to achieve the same outcome as a Stackelberg solution, based on the available information. Therefore, the defender can apply non-game-theoretic strategy selection techniques as discussed in the simulation results in Section \ref{Sec:results}.  

In our formulated game, the cost functions in \eqref{AttackersAccUtility} and \eqref{DefendersAccUtility} will ensure the existence of the Stackelberg solution, under closed-loop perfect information~\cite{basar1999dynamic}. However, this solution might not be unique, as there might be multiple strategies that yield the same utilities for the players. Solving for closed-loop strategies, in general, is challenging, especially when the number of time steps is large. In the formulated game, the number of available actions for each player, at every time step, equals $5$ which is the number of the UAVs. As a strategy is a combination of $\tau$ different actions, there will be $5^\tau$ different strategies available for each player. The solution follows by calculating the attacker's response for each of $5^\tau$ different defender's strategies, which involves testing all the attacker's $5^\tau$ strategies per a defender's strategy. Finally, the defender selects the pair of strategies that minimize its utility. The complexity of this solution approach will then be $\mathcal{O}(5^{2\tau})$, which is exponential in terms of the number of time steps. This, in fact, might not be feasible when the value of $\tau$ is large, as is the case in the UAVs' traveling model. To this end, we propose a computationally efficient solution of the game as shown in the next theorem.

\begin{theorem}\label{Theorem2}
The solution of the closed-loop dynamic Stackelberg game $\mathit{\Xi}$ is equivalent to solving the static Stackelberg equilibrium at each individual time step.
\end{theorem}
\begin{proof}
The proof is presented in Appendix \ref{theorem2proof}.
\end{proof}

From Theorem \ref{Theorem2}, we can infer that the complexity of obtaining the solution will be reduced to determining the attacker's response and the defender's Stackelberg action at each time step. Thus, the complexity of the game will be reduced to $\mathcal{O}(5^2\tau)$, which is linear in terms of the number of time steps. Note that, as discussed earlier, the solution of the formulated game is non-unique. Theorem \ref{Theorem2} will then allow obtaining one of the Stackelberg equilibrium strategies.


\section{Simulation Results and Analysis}\label{Sec:results}

For our simulations, we consider the case of one GPS spoofer and one group of five UAVs to better highlight the outcome of the dynamic game. The analysis will can apply to multiple groups of UAVs with an expected better outcome for the UAVs due to the decreased probability of attacking a single UAV.

\begin{figure*}
	\centerline{\includegraphics[width=\paperwidth]{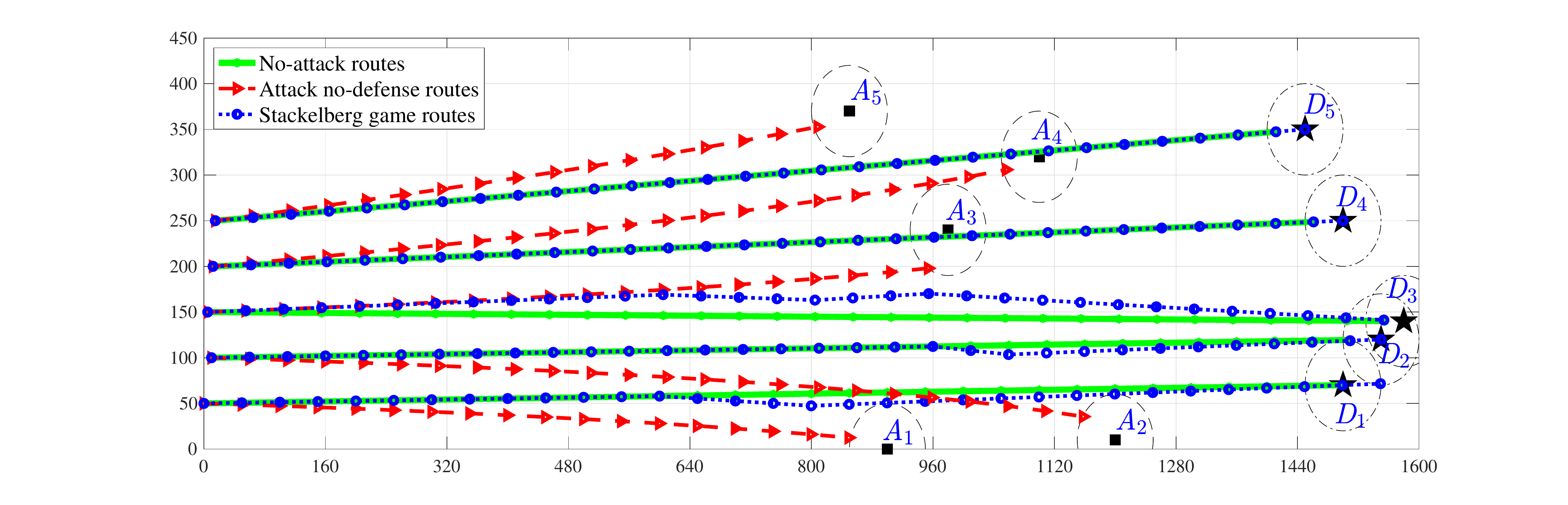}}
	\caption{UAVs routes under no attack, no defense, and under the proposed Stackelberg game solution.}\label{fig:visoutput}
	\vspace{0.1cm}
\end{figure*}

First, in Fig.~\ref{fig:visoutput}, we show a visual output of the proposed game. The UAVs' starting locations, real destinations, and the attacker's desired destinations are all shown in this figure. The points $A_i$, $i=1,\dots,5$ are the attacker's desired destinations for each UAV and the points $D_i$, $i=1,\dots,5$ are the real destinations for each UAV. The UAVs update their locations every $50$ meters and  the value of $e_{\textrm{max}}$ is assumed to be $50$~m as well. Note that the maximum value of $e_{\textrm{max}}$ that keeps the attack covert depends on the fault detector used within the UAV~\cite{su2016stealthy} and it can range from few meters up to $90$~m for different detector settings.

Fig.~\ref{fig:visoutput} shows the different routes that each UAV can follow. The straight lines connecting each UAV's starting point to its destination, $D_i$, represent the shortest paths that each UAV will follow if there is no attack. These routes result from calculating the UAVs' locations according to \eqref{eq:optimalNextLocation} at each time step.
On the other hand, the routes from the UAVs' starting points to the attacker's desired destinations, $A_i$, are the routes resulting from following the attacker's imposed locations at each time step, i.e., the locations in \eqref{eq:nextLocationAttack}. Note that, these routes, unlike the shortest paths, are not straight lines as they are composed of multiple short segments each of which is the UAV path after perceiving the attacker's imposed location. In some of these paths, the attacker imposes multiple locations along the path causing the route to deviate towards the attacker's desired destinations. The UAVs can follow these routes only if there is an attack while the defense mechanism is not used.
Finally, Fig.~\ref{fig:visoutput} also shows the routes resulting from the proposed Stackelberg game solution. These routes are bounded by the previous two routes and may coincide with parts of these routes. Every change in these routes represents a change in the attacker's response action, and, hence a change in the UAV under attack. In the following, we will study how the routes resulting from the Stackelberg game compare to the previous two sets of routes, i.e., routes under no attack and routes under attack while no defense is used.

%
%

Next, we study the effect of GPS spoofing attacks on the UAVs' traveling routes.

\begin{figure}[t]
	\centerline{\includegraphics[width=7.5cm]{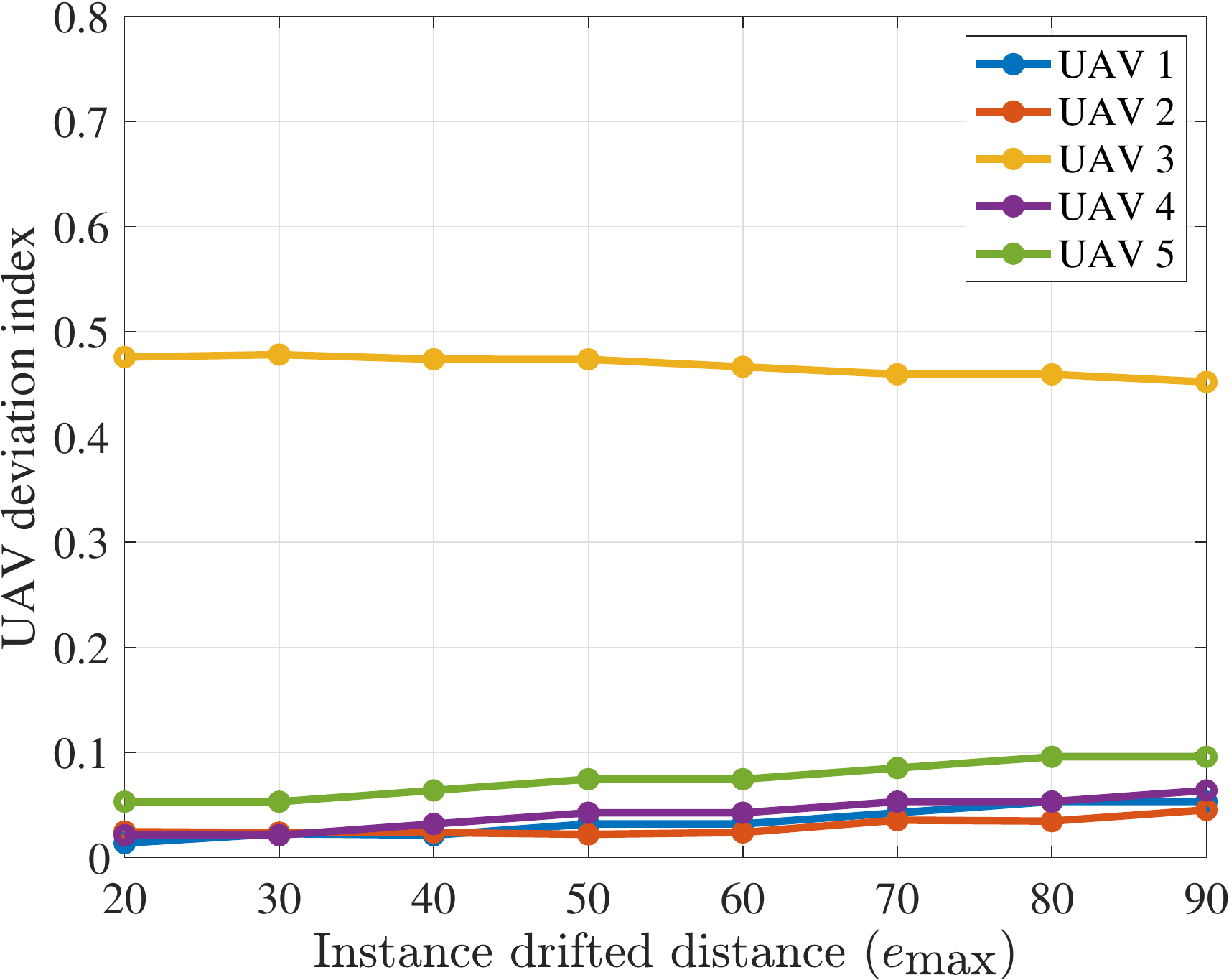}}
	\caption{UAVs deviation index as a relation in the instance drifted distance.}\label{fig:resilience1}
\end{figure}

\subsection{UAVs Deviation due to Spoofing Attacks}
To study the attacker's effect on the UAVs' traveling routes, we define a deviation metric, for each UAV, to compare the UAV's route resulting from the Stackelberg game with both the no-attack route and the attack no-defense route shown in Fig. \ref{fig:visoutput}. Let $\boldsymbol{x}_i^r(t)$ be UAV $i$'s location on the expected route under no attack and let $\boldsymbol{x}_i^f(t)$ represent its locations on the attacker's desired route. We then define $\theta_i(t)$ to be the deviation of UAV $i$ at time step $t$ given by:
\begin{equation}
\theta_i(t) = 1-  \frac{\Bigl|\!\Bigl|  \boldsymbol{x}_i(t) - \boldsymbol{x}^f_i(t)  \Bigr|\!\Bigr|^2_2}{\Bigl|\!\Bigl|  \boldsymbol{x}_i^r(t) - \boldsymbol{x}^f_i(t)  \Bigr|\!\Bigr|^2_2}, 
\end{equation}
such that when a UAV is traveling on a no-attack route, the value of $\theta_i(t)$ will equal $0$. Similarly, if a  UAV is traveling on the attacker's desired route, the value of $\theta_i(t)$ will equal $1$. Any other value in between the two routes, the value of $\theta$ will be $0 < \theta < 1$.  We can, then, define the deviation index $\Theta_i$ for UAV $i$ as the average of its deviation over all time steps. Note that the deviation index can capture how far each UAV has traveled from its planned route towards the attacker's desired route. However, a higher deviation index does not necessary mean that this specific UAV will be captured by the attacker. It merely means that the attacker has disrupted the UAV's original route.

In Fig. \ref{fig:resilience1}, we study the effect of the instance drifted distance, $e_{\textrm{max}}$, on the UAVs' deviation indices. To better highlight the effect of $e_{\textrm{max}}$, we allowed the UAVs to update their locations frequently by setting the update distance to $15$ m. We notice that, as $e_{\textrm{max}}$ increases, the attacker will be able to induce bigger changes to the UAVs' locations causing them to deviate more from the planned routes hence increasing their deviation index. For instance, when $e_{\textrm{max}}= 20$~m, some of the UAVs have almost zero deviation from their planned routes. On the other hand, when $e_{\textrm{max}}= 90$~m, most UAVs have a slight deviation  from their planned routes. We also note, from Fig. \ref{fig:resilience1}, that UAV $3$, has much higher deviation than the other UAVs. This happens as UAV $3$ affects the attacker's utility the most while having a smaller effect on the defender's utility. Thus, UAV $3$ is chosen by the attacker, at most time steps, as its response action while the defender chooses other UAVs to protect. Note that the players' utilities, and, hence, their chosen actions (UAVs) depend on the UAVs' current locations and both the real and the attacker's desired destinations.

\begin{figure}[t]
	\centerline{\includegraphics[width=7.5cm]{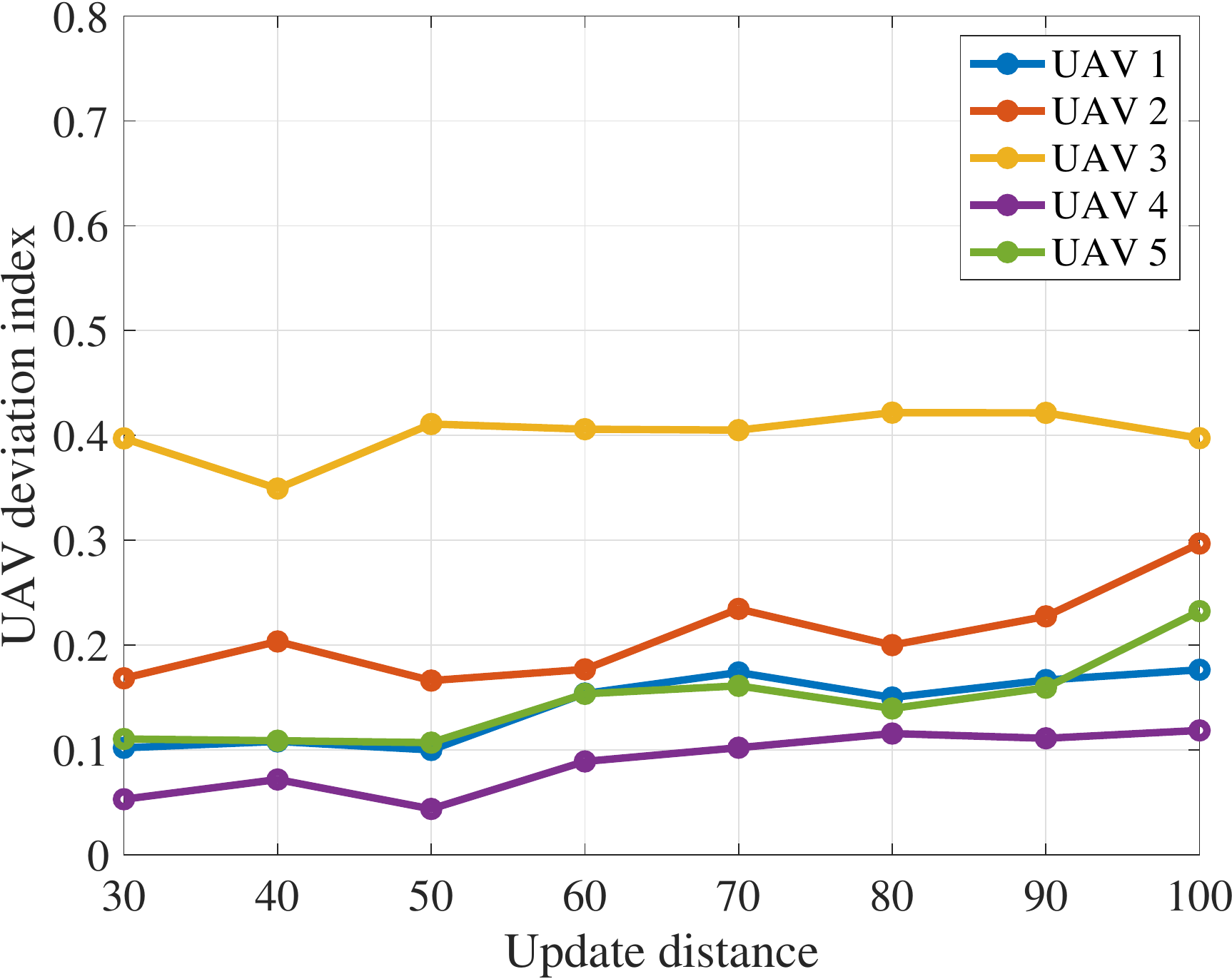}}
	\caption{UAVs deviation index as a relation in the update distance.}\label{fig:resilience2}
\end{figure}

Next, we study the effect of the update distance, i.e., the distance at which the UAVs apply the defense mechanism, on the UAVs' deviation indices. As the UAVs are traveling using $u^{\textrm{max}}$, the update distance will indicate the frequency of updating the UAVs' locations. In Fig. \ref{fig:resilience2}, we study different update distances on the deviation index. In this case, we set the value of $e_{\textrm{max}}$ to $60$ m. From Fig. \ref{fig:resilience2}, we can see that, when the update distance increases, i.e., the less frequent the UAVs apply the defense mechanism, the more they deviate from their planned routes. For instance, when the update distance is set to $30$ m, the average deviation index of all UAVs is $0.17$ compared to $0.25$ when the update distance is $100$. This happens as the UAVs will travel more towards the attacker's destinations before they update their locations, and, move towards their real destinations. We also note that changing the update distance can change the effect of the UAVs on the players' utilities, and, hence on their actions. For instance, when the update distance is $40$ m, UAV $2$ is attacked more than when the update distance is $30$ m, and, hence, it has a higher deviation index. For the same update values, UAV $3$ has a lower deviation index when the update distance is $40$ m compared to when it is $30$ m.

%

\begin{figure}[t]
	\centerline{\includegraphics[width=7.5cm]{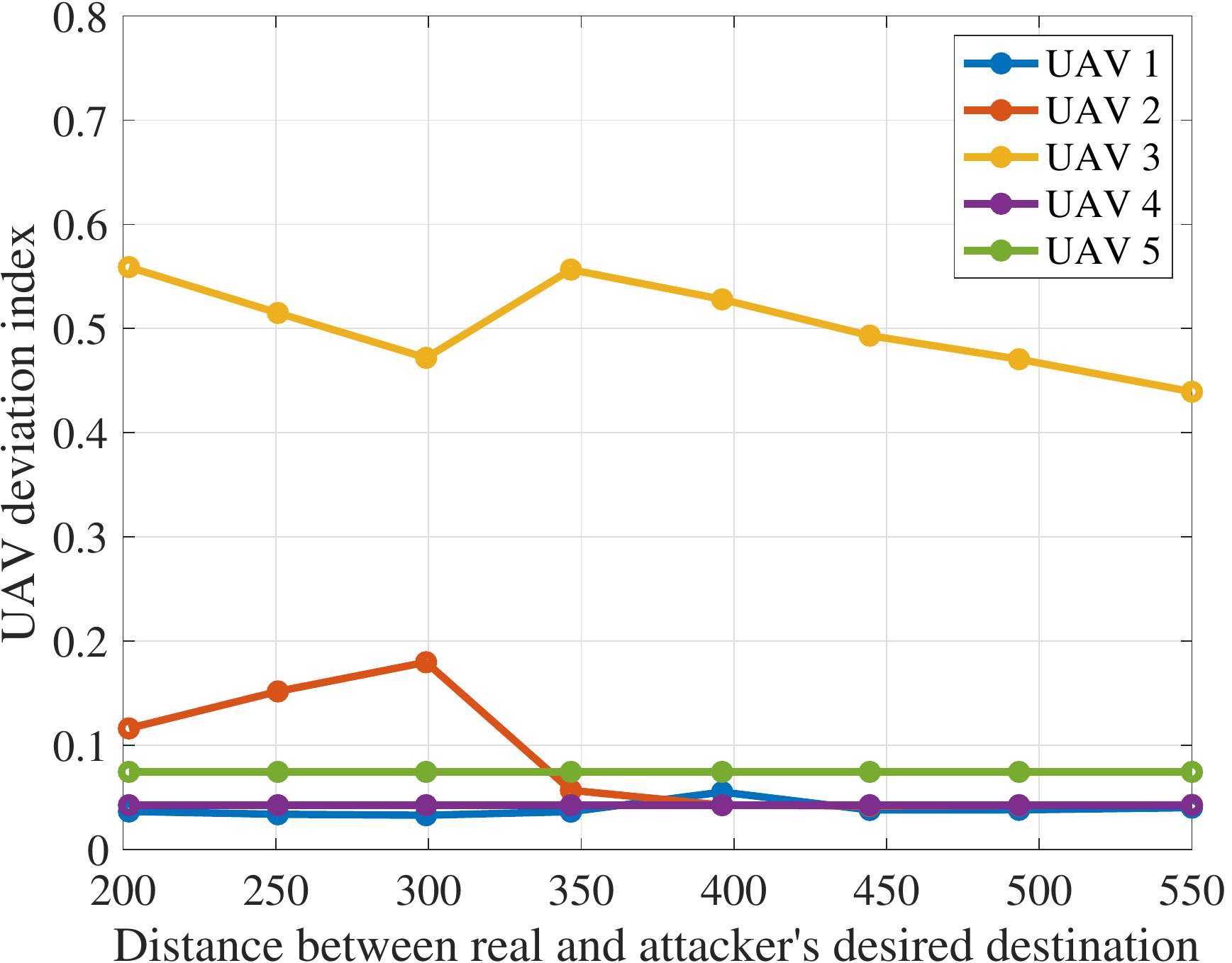}}
	\caption{UAVs deviation index change due to shifting the attacker's desired destinations.}\label{fig:resilience3}
\end{figure}

\begin{figure*}
	\begin{subfigure}{.3\textwidth}
		\centering
		\includegraphics[width=6.2cm]{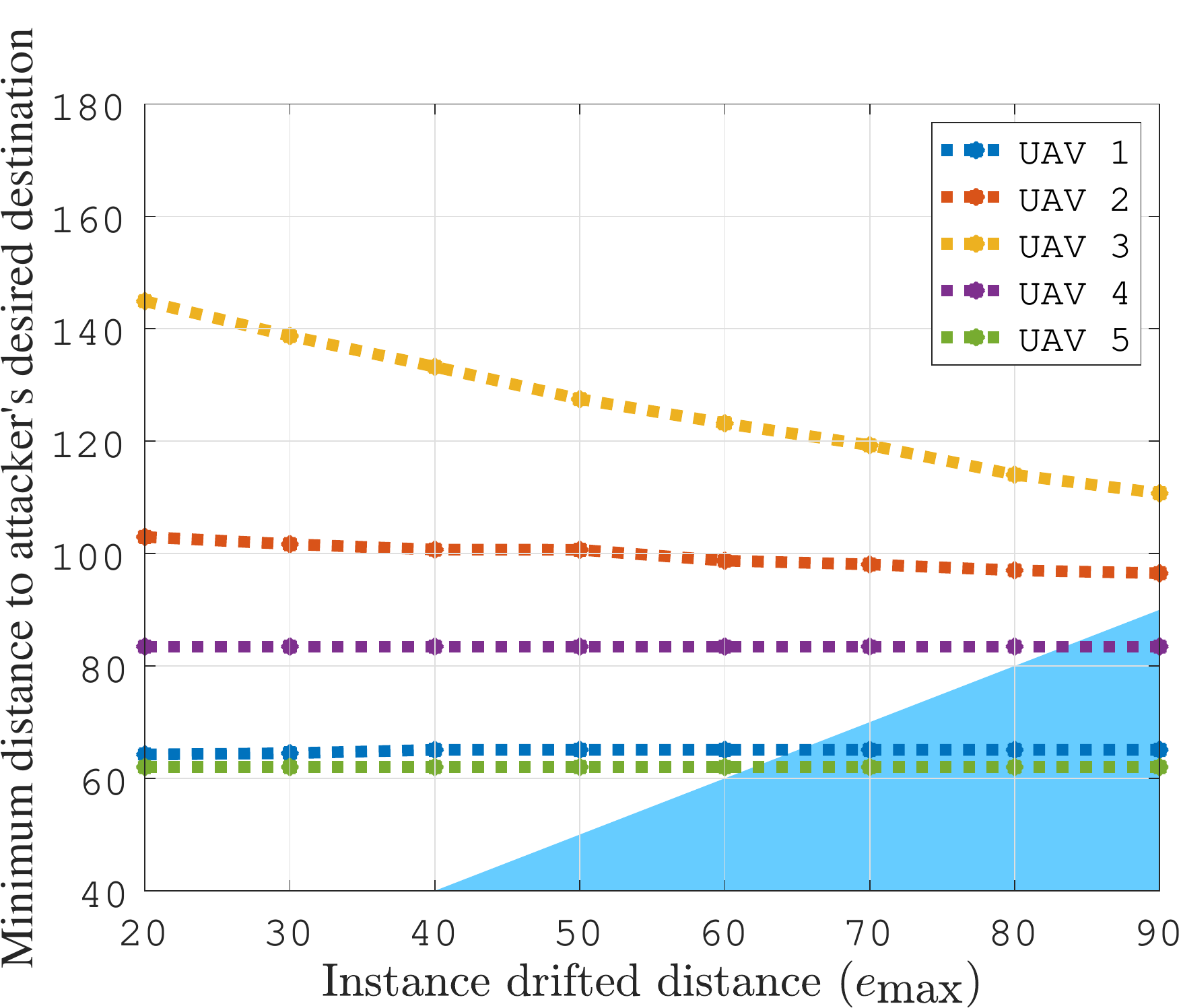}
		\caption{Stackelberg strategies}\label{fig:capture1_S}
	\end{subfigure}\hfill
	\begin{subfigure}{.3\textwidth}
		\centering
		\includegraphics[width=6.2cm]{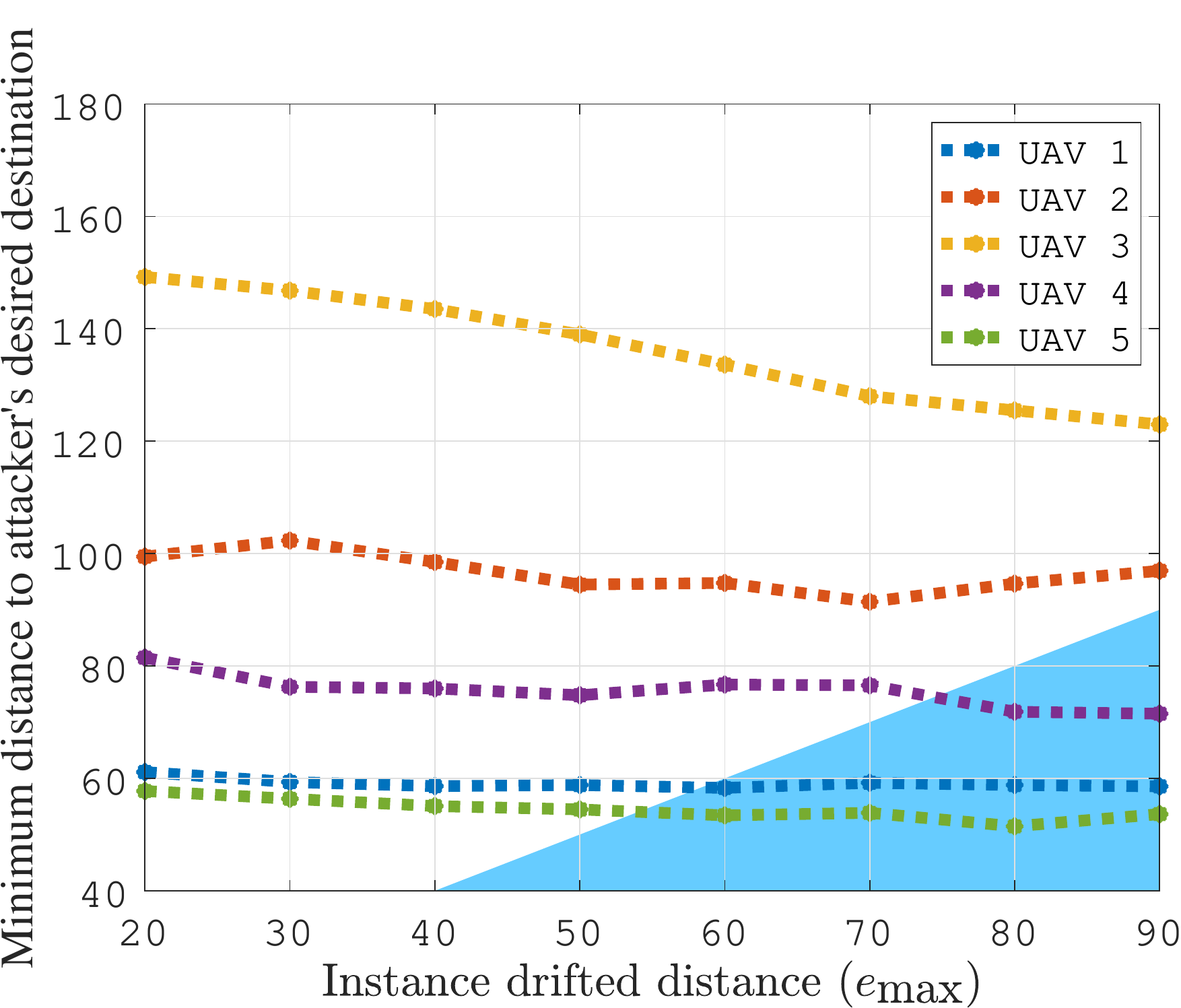}
		\caption{Random strategies}\label{fig:capture1_R}
	\end{subfigure}\hfill
	\begin{subfigure}{.3\textwidth}
		\centering
		\includegraphics[width=6.2cm]{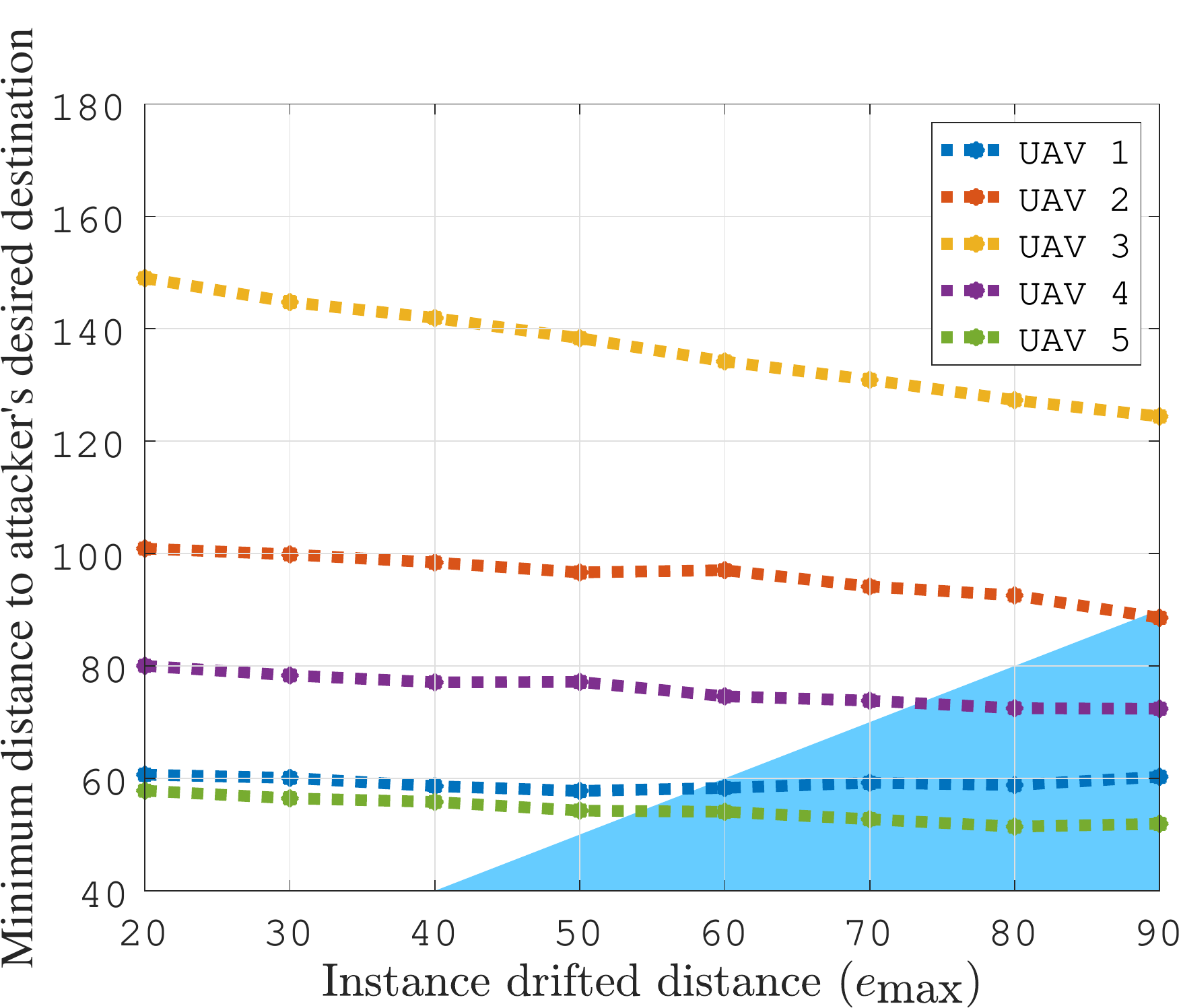}
		\caption{Deterministic strategies}\label{fig:capture1_D}
	\end{subfigure}
	\caption{The effect of the instance drifted distance, $e_{\textrm{max}}$, on the possibility of UAV capture.}\label{fig:capture1}
		\vspace{0.1cm}
\end{figure*}

Next, we study the effect of changing the attacker's desired destinations on the UAVs' deviation indices. In this case, we use the same parameters as in Fig. \ref{fig:visoutput}. Different attacker's desired destinations, $A_i$, $i=1,\dots,5$, are tested by reducing the distance between the attacker's desired destination and the real destinations randomly with an average change of $50$ meters per UAV. In this case, the real destinations are fixed, and the attacker's desired destinations are shifted along the $x$ direction only allowing for deviations to take place along the travel routes. Fig. \ref{fig:resilience3} shows the effect of changing the attacker's desired destinations on the UAVs' deviation indices.
We can see that, as the attacker's desired destinations are closer to the real destinations, the deviation index increases. For instance, when the average distance between the attacker's and the real destinations is $550$ m, the average deviation index for all UAVs is $0.13$  compared to $0.17$ when the average distance drops to $200$ m.
This is because when the distances are smaller, the attacker will have more opportunities (longer paths) to attack the UAVs causing them to deviate more from their planned routes.
Note that, as the distances between the destinations are allows to change randomly, the changes in UAVs' distances to the attacker's desired destinations will not be constant. Hence, the UAVs will contribute differently to the attacker's utilities with each change. This will cause the attacker's actions (attacked UAVs) to be different over the travel routes, with each change. For example, we can see in Fig. \ref{fig:resilience3} that when the distance is $300$ m, UAV $2$ is attacked more than when the distance is $350$ m. On the other hand, UAV $3$ is less attacked when the distance changes from $350$ m to $300$ m.

Finally, we note that, in all the previous cases, while the three studied parameters affect the deviation index of the UAVs, the update distance has the most effect on the UAVs' deviation indices. This is because delaying the update will allow the UAVs to travel on the attacked routes for longer distances before correcting their locations leading to larger deviations. Meanwhile, in the other scenarios, the attacker caused a smaller average deviation on the UAVs because the UAVs update their locations more frequently. These findings corroborate the importance of the proposed defense mechanism and provide important insights for the drone operator to choose suitable update distances according to the available resources.

Next, we study the attacker's possibility to capture any of the UAVs under GPS spoofing attacks.

\subsection{Capturing Possibilities under GPS Spoofing Attacks}

To study the capture possibility of the UAVs, we assume the attacker needs to change a UAV's route, by imposing a different location, in order to capture it. We also assume that the UAV will be captured if it reaches a distance $e_{\textrm{max}}$ from the attacker's desired destination. In the following, we will compare our proposed Stackelberg solution with two other non-game-theoretic baselines referred to as random and deterministic approaches. In the random approach, the defender chooses any UAV randomly to protect at each time step. In the deterministic approach, the defender considers all the UAVs, in order, by choosing one at each time step. In all the three cases, the attacker chooses its strategies in response to the defender's chosen strategies. Note that the non-game-theoretic strategies represent the defender's possibilities in case of incomplete information, i.e., when the defender cannot observe the attacker's actions and, hence cannot apply the game-theoretic approach.

Fig. \ref{fig:capture1} shows the minimum distance that each UAV can reach from its attacker's desired destination, for each value of $e_{\textrm{max}}$. The shaded areas in Fig. \ref{fig:capture1} represent the distances under which the UAV is considered to be captured which correspond to having a distance less than $e_{\textrm{max}}$ to the attacker's desired destination. The configuration parameters in this case are similar to Fig. \ref{fig:resilience1}. We can see in Fig. \ref{fig:capture1_S} that using the Stackelberg strategies, the defender is able to protect all the UAVs until $e_{\textrm{max}}=60$~m. When $e_{\textrm{max}}=70$~m and $e_{\textrm{max}}=80$~m, both UAVs $1$ and $5$ can be captured by the attacker and when $e_{\textrm{max}}=90$~m, UAV $4$ can also be captured. In Fig. \ref{fig:capture1_R}, when the drone operator uses random strategies, we can see that the attacker is able to capture UAVs $1$ and $5$ starting from $e_{\textrm{max}}=60$~m. Moreover, UAV $4$ can be captured starting from $e_{\textrm{max}}=80$~m. Under deterministic strategies, Fig. \ref{fig:capture1_D}, the attacker is also able to start capturing UAVs $1$ and $5$ starting from $e_{\textrm{max}}=60$~m. It will be also able to capture UAV $4$ starting from $e_{\textrm{max}}=80$~m and to capture UAV $2$ at $e_{\textrm{max}}=90$~m.

It is clear from Fig. \ref{fig:capture1} that following the Stackelberg strategies will help the defender to protect more UAVs, particularly for higher values of $e_{\textrm{max}}$. This is due to the fact that, under Stackelberg strategies, the drone operator considers its utilities based on the attacker's response strategies which allows it to mitigate the effect of the attacker's expected actions. However, the other approaches can still help the defender to utilize the defense mechanism and protect the UAVs to some extent, when the defender cannot observe the attacker's actions and use the game-theoretic approach. Fig. \ref{fig:capture1} also shows that UAV $3$ has the most changes to its minimum distance from the attacker's desired destination, when $e_{\textrm{max}}$ increases. This corroborates the result of Fig. \ref{fig:resilience1} whereby UAV $3$ was the most affected by the spoofer's imposed locations, in terms of deviation from its planned route. However, UAV $3$ remains far enough from being captured as the defender's actions allow it to return to the correct traveling direction.

\begin{figure*}
	\begin{subfigure}{.3\textwidth}
		\centering
		\includegraphics[clip,width=6.2cm]{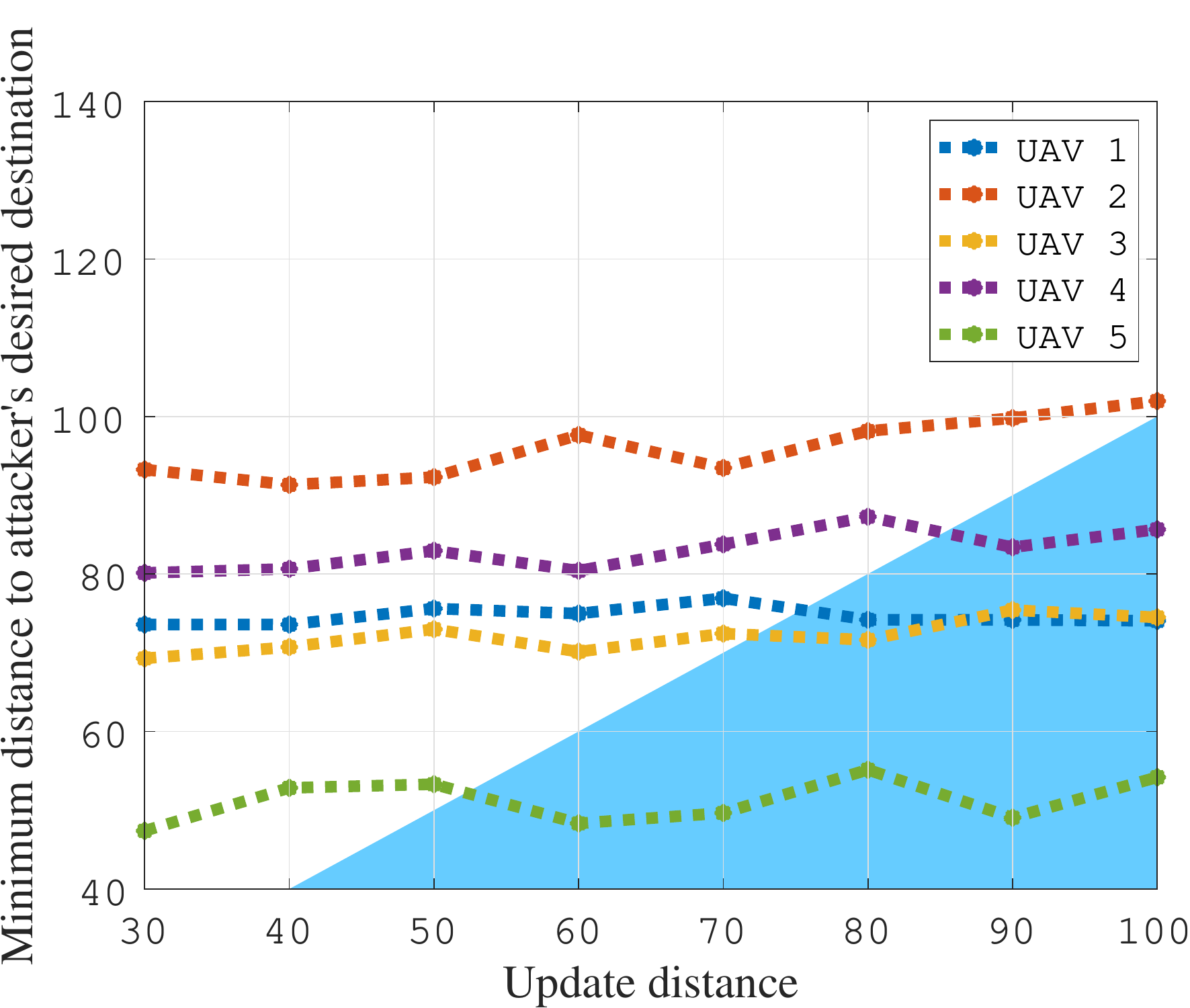}
		\caption{Stackelberg strategies}\label{fig:capture2_S}
	\end{subfigure}\hfill
	\begin{subfigure}{.3\textwidth}
		\centering
		\includegraphics[clip,width=6.2cm]{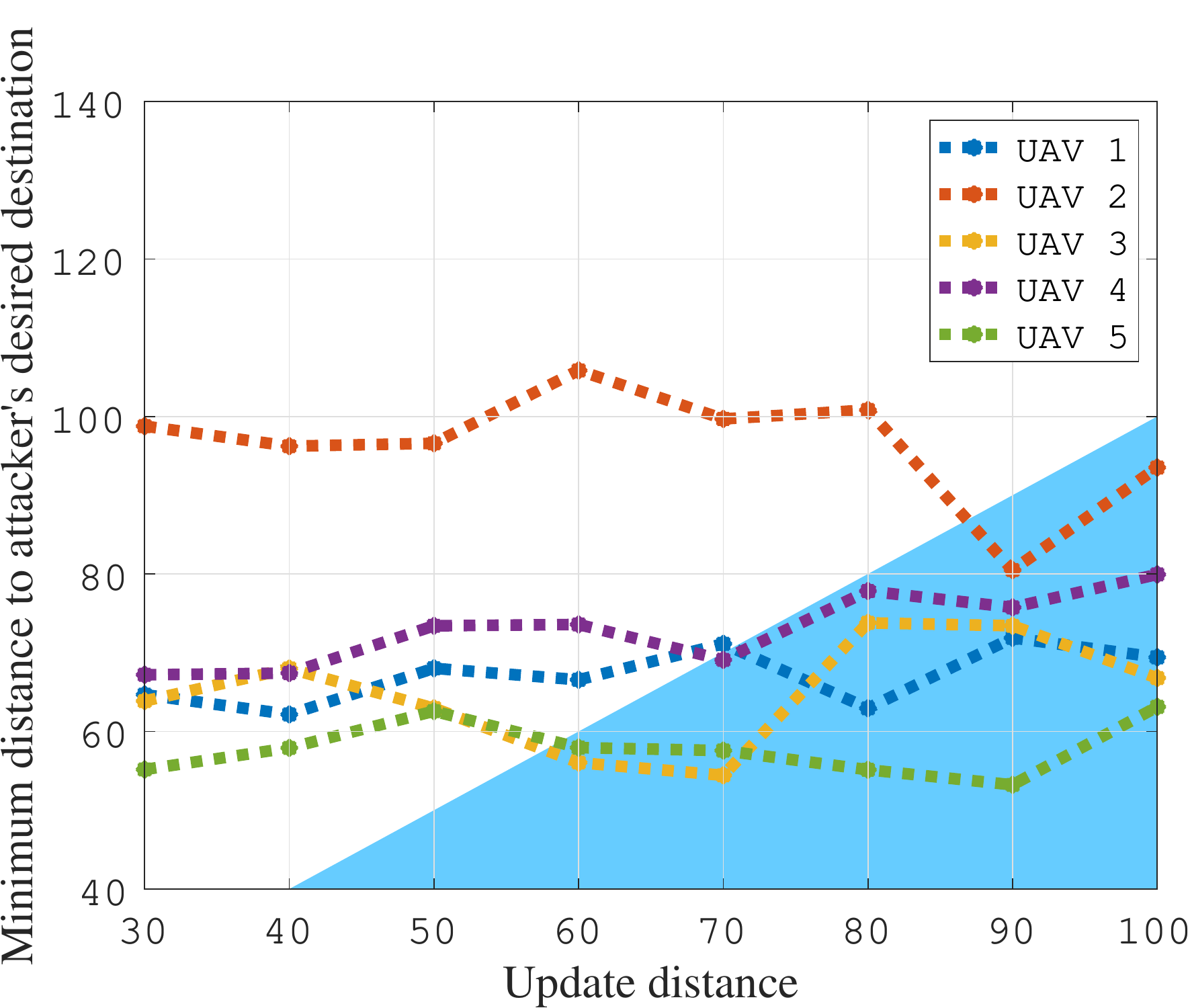}
		\caption{Random strategies}\label{fig:capture2_R}
	\end{subfigure}\hfill
	\begin{subfigure}{.3\textwidth}
		\centering
		\includegraphics[clip,width=6.2cm]{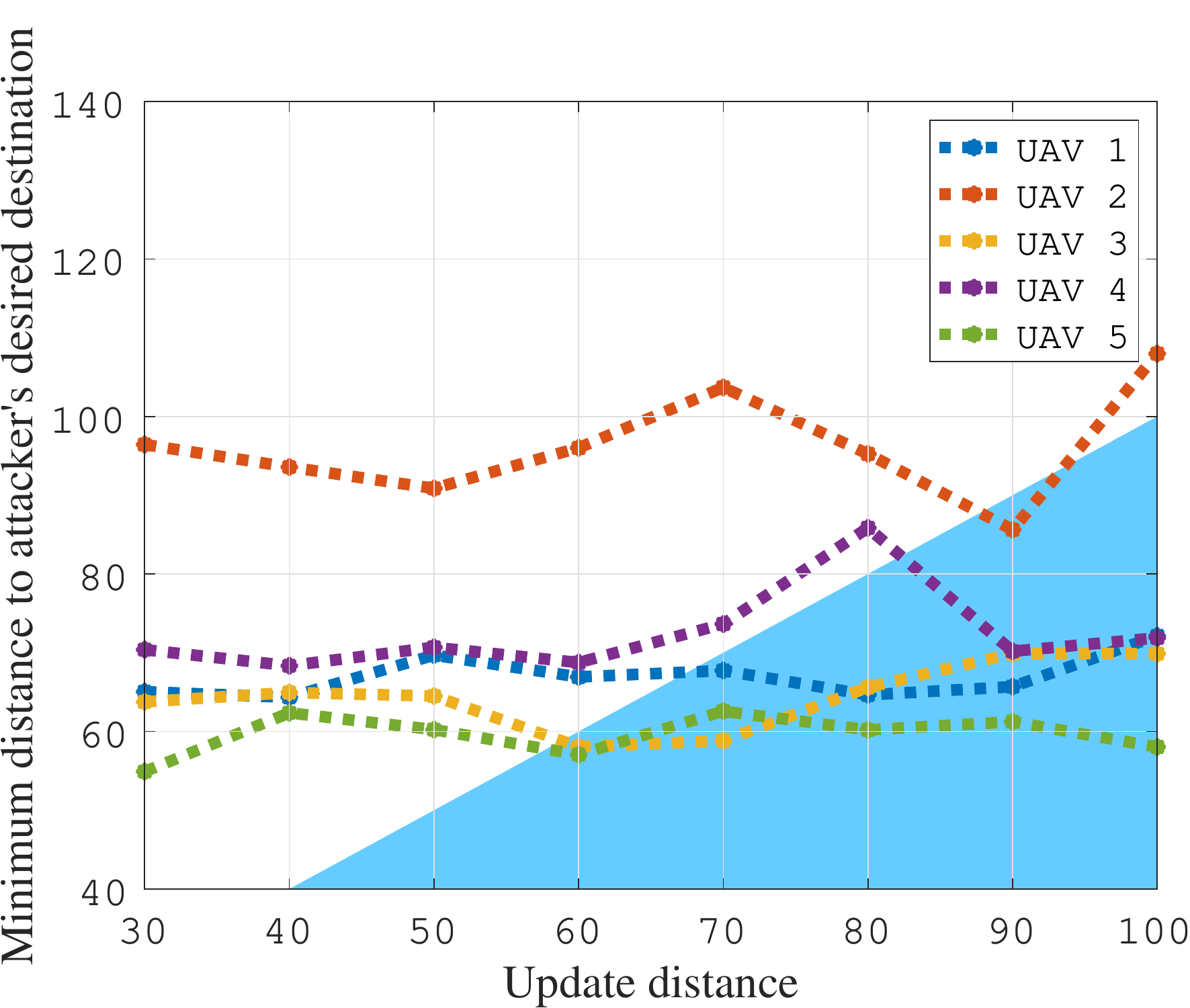}
		\caption{Deterministic strategies}\label{fig:capture2_D}
	\end{subfigure}
	\caption{The effect of update distance on the possibility of UAV capture.}\label{fig:capture2}
		\vspace{0.1cm}
\end{figure*}

Next, we study the effect of changing the UAVs' update distance on the possibility of UAV capture. The configuration parameters in this case are similar to Fig. \ref{fig:resilience2}. Fig. \ref{fig:capture2} shows the effect of the capture possibility with the shaded areas representing the distances under which the UAVs are considered to be captured. We can see in Fig. \ref{fig:capture2_S} that using the Stackelberg strategies, the drone operator is able to protect all the UAVs until an update distance of $50$m. The attacker is able to capture its first UAV, UAV $5$, when the update distance is $60$~m or $70$ m. When the update distance is $80$~m, three UAVs can be captured by the  attacker, and four UAVs can be captured for update distances greater than $80$ m. Under random strategies, Fig. \ref{fig:capture2_R} shows that the attacker is able to capture more UAVs when the update distance is $60$ m as it will capture UAVs $3$ and $5$. Similarly, for all the consequent update distances, more UAVs can be captured compared to the Stackelberg strategies. Under deterministic strategies, Fig. \ref{fig:capture2_D} shows that the attacker will be able to start capturing UAVs $3$ and $5$ at an update distance of $60$~m. When the update distance is $70$ m, the attacker will be able to capture three UAVs compared to one in the Stackelberg strategies. For the consequent update distances, the attacker is able to capture at least the same number of UAVs as the Stackelberg strategies.

Next, we study the effect of changing the attacker's desired destinations on possibility of UAV capture. The configuration parameters in this case are similar to Fig. \ref{fig:resilience3}, with $e_{\textrm{max}}=60$~m. Fig. \ref{fig:capture3} shows this effect on the capture possibility. The shaded areas in Fig. \ref{fig:capture3} represent the capture distances of the UAVs, i.e., distances less than $60$~m. We can see in Fig. \ref{fig:capture3_S} that using the Stackelberg strategies, the drone operator is able to protect all the UAVs when the average distance is $550$ m. When the average distance decreases to $500$ m, the attacker will be able to capture UAV $5$. For any distances less than $400$ m, the attacker will be able to capture UAVs $1$ and $5$. Under random strategies, Fig. \ref{fig:capture3_R}, we can see that the attacker is able to capture UAVs $1$ and $5$ for all the considered distances. The attacker will also be able to capture UAV $4$ under multiple distance settings. Finally, we can see in Fig. \ref{fig:capture3_D}, that the deterministic strategies show very similar response to the random strategies, in terms of the possibility of UAV capture.

Note that, from the three previous scenarios, we can see that the update distance has the most effect on the possibility of UAV capture, similar to its effect of the deviation index. This highlights the importance of choosing this critical parameter when applying the proposed defense mechanism.

\section{Conclusion}\label{Sec:conclusion}
In this paper, we have proposed a novel framework to mitigate the effects of capture attacks via GPS spoofing that target UAVs. Systems dynamics have been used to model the UAVs' optimal routes towards their destinations. To study the effect of a GPS spoofer on these optimal routes, we have mathematically derived the spoofer's optimal imposed locations on any UAV. These locations, when imposed on a UAV, cause the UAVs to deviate from their planned routes and follow new routes towards the spoofer's desired destinations.
We have then proposed a countermeasure defense mechanism to allow UAVs to determine their real locations, after being attacked. This countermeasure is built on the premise of cooperative localization, in which a UAV uses the locations of nearby UAVs to determine its real location. We have, then, defined a Stackelberg game problem to allow the UAVs to better utilize the proposed defense mechanism. In particular, the game is formulated between a GPS spoofer and a drone operator that manages a number of UAVs. The drone operator is considered the leader that determines its strategies first and the spoofer then responds by choosing its strategies. We have mathematically derived the Stackelberg equilibrium strategies, for the formulated game, through a computationally efficient approach. Results have shown that the proposed defense mechanism along with Stackelberg equilibrium strategies outperform other strategy selection techniques in terms of reducing the possibility of UAV capture. We have also tested the effect of different parameters on the UAVs' deviation indices and on the possibility of UAV capture and the results have shown that the UAV update distance has the most effect on these metrics. For future work, we will consider the case of protecting groups of more than five UAVs against multiple simultaneous GPS spoofing attacks.
 
\begin{figure*}
	\begin{subfigure}{.3\textwidth}
		\centering
		\includegraphics[clip,width=6.2cm]{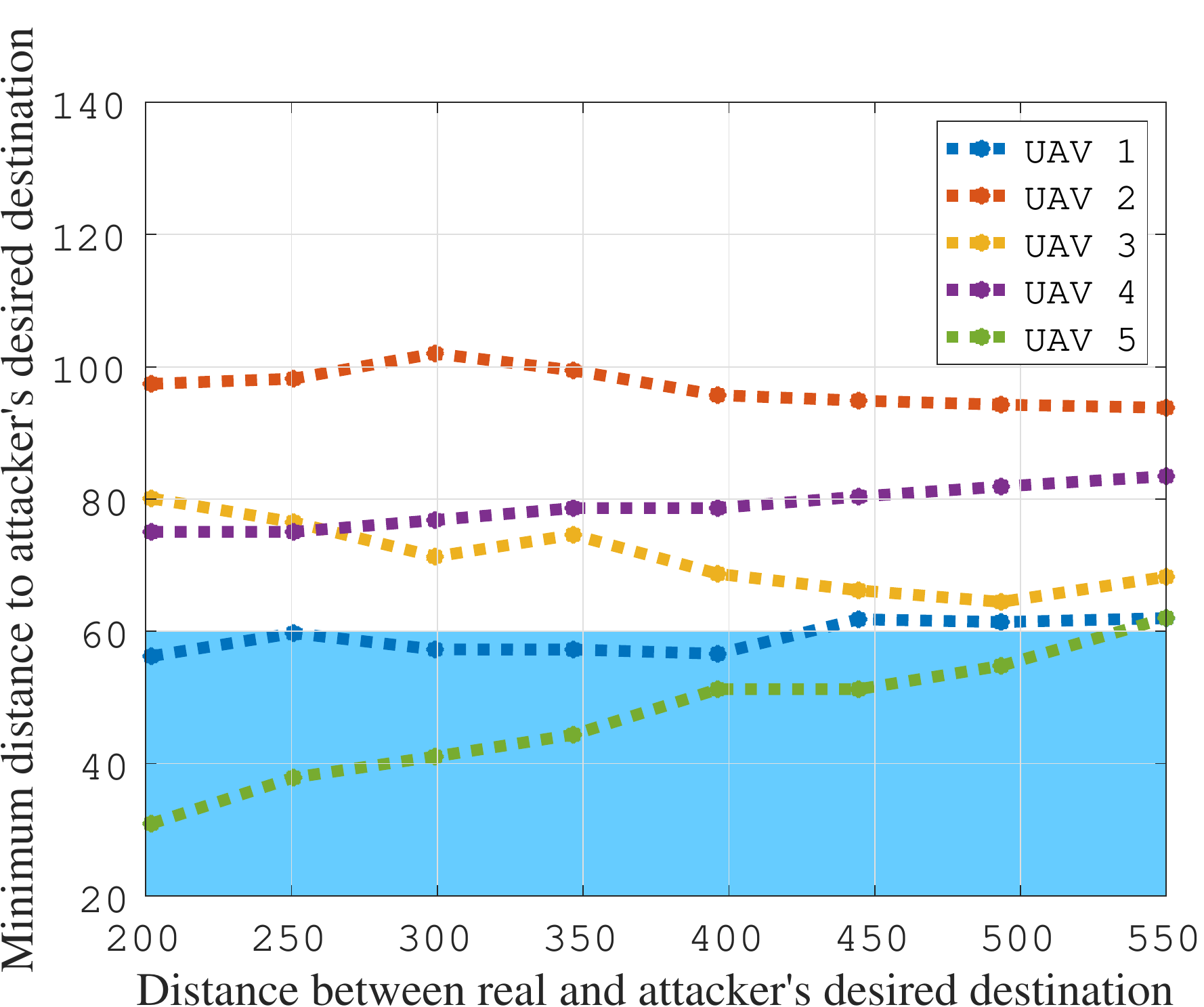}
		\caption{Stackelberg strategies}\label{fig:capture3_S}
	\end{subfigure}\hfill
	\begin{subfigure}{.3\textwidth}
		\centering
		\includegraphics[clip,width=6.2cm]{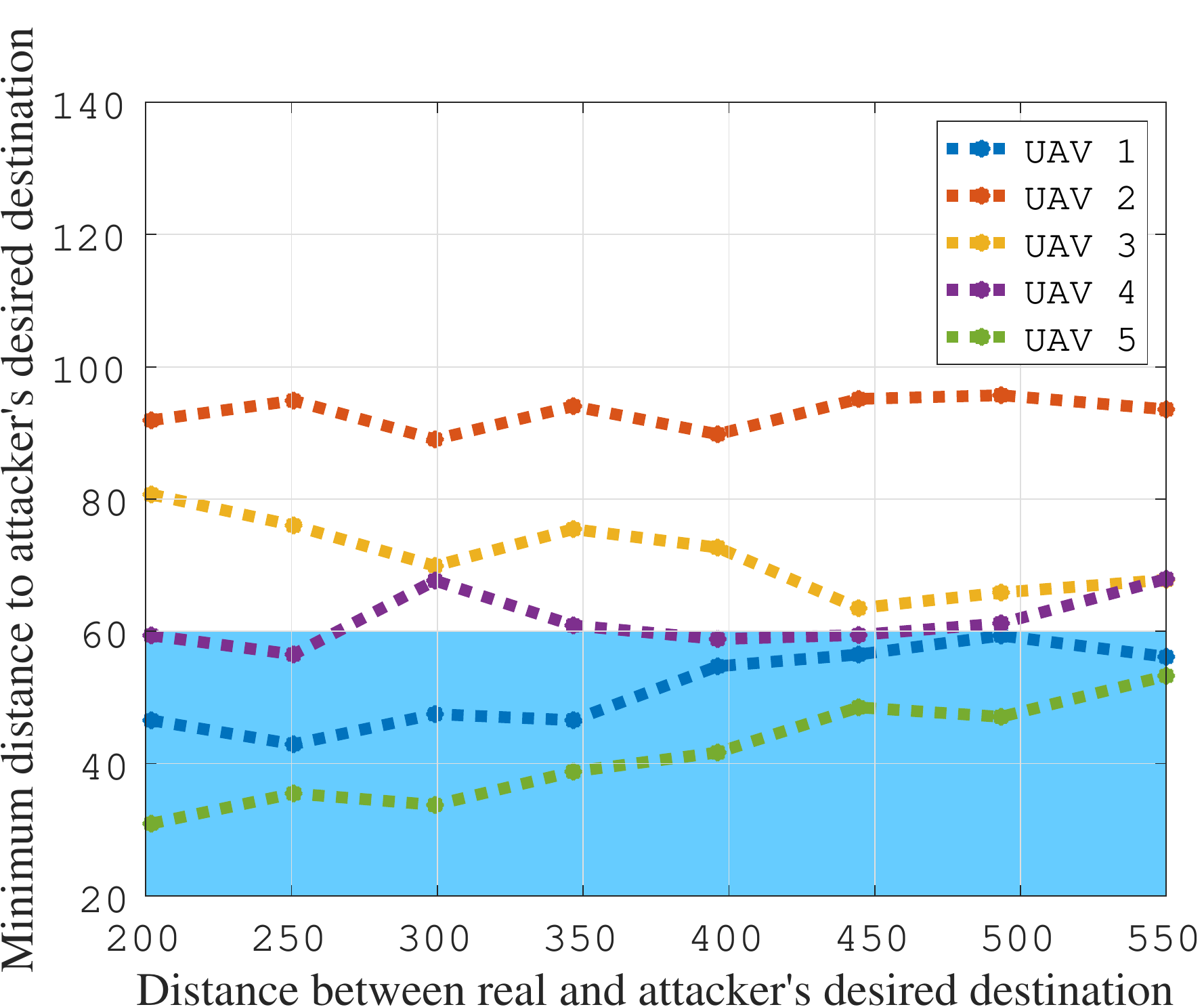}
		\caption{Random strategies}\label{fig:capture3_R}
	\end{subfigure}\hfill
	\begin{subfigure}{.3\textwidth}
		\centering
		\includegraphics[clip,width=6.2cm]{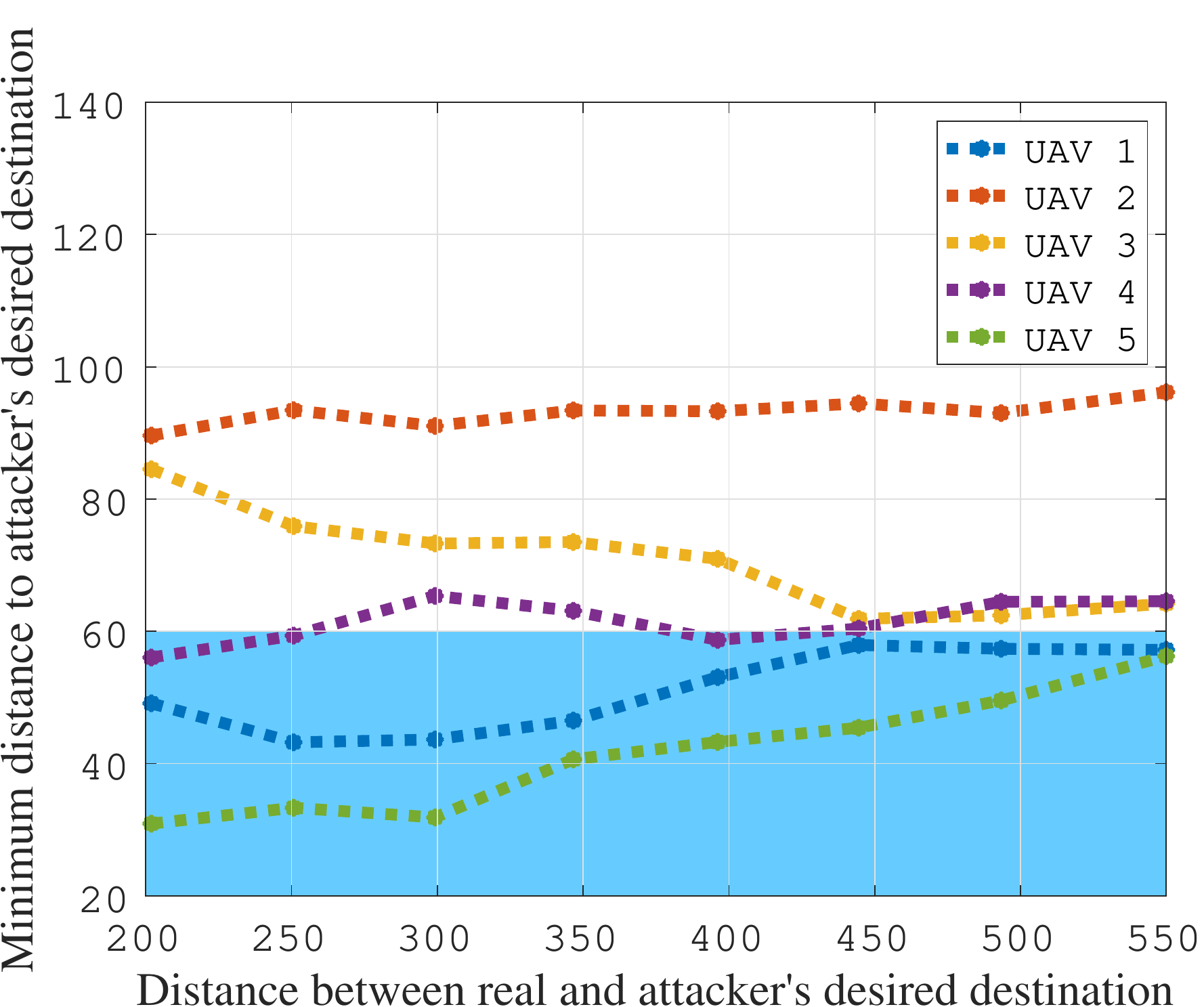}
		\caption{Deterministic strategies}\label{fig:capture3_D}
	\end{subfigure}
	\caption{The effect of the average distance between the real and the attacker's desired destinations on the possibility of UAV capture.}\label{fig:capture3}
\end{figure*}

\begin{appendices}

\section{Proof of Proposition 1}\label{Prop1proof}

\begin{proof} We start by making two assumptions. First, the UAV is considered to reach its destination if it is within a distance $e_{\textrm{max}}$ from its real destination. We also assume that the distance between the UAV's current location and its real destination is greater than $e_{\textrm{max}}$, i.e., the UAV did not reach its destination yet.
	
	According to (\ref{eq:optimalNextLocation}) and (\ref{eq:nextLocationAttack}), the UAV travels towards its real destination from its perceived location. Therefore, in order for the UAV to reach a destination in the opposite side from its real destination, the UAV needs to change its direction in the direction with the longest difference from the UAV's current location. Assume without loss of generality that the difference in the $x$ direction, between the UAV's perceived location and its real destination, is bigger than the difference in the $y$ direction. Then for the UAV, to change its $x$ direction, the value of $u_{x_i}^a(t)$ needs to flips its sign in (\ref{eq:nextLocationAttack}). Comparing the optimal controller, in $x$ direction, in both \eqref{eq:forcecomponentshat} and \eqref{eq:forcecomponents} and assuming they have opposite signs:
	\begin{equation*}
	\cos(\arctan\Big(\frac{y_{d_i} -\hat{y}_i(t)}{x_{d_i} - \hat{x}_i(t)}\Big)) 
	= - \cos(\arctan\Big(\frac{y_{d_i} -y_i(t)}{x_{d_i} - x_i(t)}\Big))
	\end{equation*}
	
	For this condition to hold, $x_{d_i} - \hat{x}_i(t)$ needs to have a different sign from $x_{d_i} - x_i(t)$. i.e., $\hat{x}_i(t)-x_i(t) = (x_{d_i}-x_i(t))+(x_{d_i} - \hat{x}_i(t))$
	
	However, under a covert attack, the imposed location is limited by (\ref{eq:normDistance}), i.e., $|\hat{x}_i(t)-x_i(t)| \le e_{\textrm{max}}|$. Since, it is assumed that the UAV's real location is more than $e_{\textrm{max}}$ away from its real destination, i.e, $x_{d_i}-x_i(t) \ge e_{\textrm{max}}$, then the condition for changing the direction cannot hold. In this case, the attacker cannot impose a location that forces the UAV to change its $x$ direction. Therefore, the attacker's desired destination cannot be in the opposite $x$ direction from the UAV's real destination. 
\end{proof}

\section{Proof of Theorem 2}\label{theorem2proof}
\begin{proof}
	We begin the proof by investigating the solution of the closed-loop dynamic Stackelberg game which is the pair of strategies from (\ref{eq:reactionset}) that satisfy (\ref{eq:StackelbergSolution}). Our proof will show that the same reaction set in (\ref{eq:reactionset}) can be achieved by considering the solution of the static Stackelberg game at each time step. Note that, the reaction set in (\ref{eq:reactionset}) is a combination of the attacker's reactions to every single defender's strategy calculated from (\ref{eq:attackerResponse}). In the following, we will show the solution when $\tau = 2$ and then generalize it to any number of time steps.
	
	When $\tau =2$, the attacker's cost function in (\ref{AttackersAccUtility}) can be written as:
	\begin{align}\label{twoStepCost}
	J^a(\boldsymbol{\beta}^d,\boldsymbol{\beta}^a)&= \sum_{i=1}^{5}  \left\lVert  \boldsymbol{x}^a_{d_i} - \boldsymbol{x}_i(\Delta,z^d_i(1),z^a_i(1)) \right \rVert^2_2 \nonumber\\
	&+  \sum_{i=1}^{5}  \left\lVert  \boldsymbol{x}^a_{d_i} - \boldsymbol{x}_i(2 \Delta,z^d_i(2),z^a_i(2)) \right \rVert^2_2.
	\end{align}
	
	In the dynamic Stackelberg game, the attacker will select a strategy $\boldsymbol{\beta}^a = \{z^a_i(1),z^a_i(2)\}$ in response to every $\boldsymbol{\beta}^d$ that minimizes its utility. We can rewrite (\ref{twoStepCost}) by substituting the values from (\ref{NextLocationActionsSimplified}):
	
	\begin{align} \label{twoStepCostActions}
	J^a(\boldsymbol{\beta}^d,\boldsymbol{\beta}^a)  &= \sum_{i=1}^{5}  \Bigl|\!\Bigl|  \boldsymbol{x}^a_{d_i} - \big( 1-z^a_i(1)+ z^d_i(1) \cdot z^a_i(1)  \big) \cdot \boldsymbol{x}_i(\Delta) \nonumber\\
	&\hspace{1cm}-  \big(z^a_i(1) -z^a_i(1) \cdot z^d_i(1)\big)   \cdot \boldsymbol{x}^a_i(\Delta)  \Bigr|\!\Bigr|^2_2  \nonumber\\
	&+  \sum_{i=1}^{5}  \Bigl|\!\Bigl|  \boldsymbol{x}^a_{d_i} - \big( 1-z^a_i(2)+ z^d_i(2) \cdot z^a_i(2)  \big) \cdot \boldsymbol{x}_i(2\Delta) \nonumber\\
	&\hspace{1cm}-  \big(z^a_i(2) -z^a_i(2) \cdot z^d_i(2)\big) \cdot \boldsymbol{x}^a_i(2\Delta) \Bigr|\!\Bigr|^2_2.
	\end{align}
	
	Now, consider that the defender chooses a specific strategy $\boldsymbol{\beta}^d=\{z^d_j(1)=1,z^d_k(2)=1\}$ where $j,k \in \mathcal{N}$. This means in the first time step $z^d_j(1) =1$ and all the remaining actions will be be zero. Similarly, in the second time step $z^d_k(2) =1$ and all the remaining actions will be be zero. The previous cost function can then be written as:
	\small
	\begin{align} \label{twoStepCostActions2}
	J^a(\boldsymbol{\beta}^d,\boldsymbol{\beta}^a)&   = \sum_{i=1,i\ne j}^{5}  \Bigl|\!\Bigl|  \boldsymbol{x}^a_{d_i} - \big( 1- z^a_i(1) \big) \cdot \boldsymbol{x}_i(\Delta) - z^a_i(1)   \cdot \boldsymbol{x}^a_i(\Delta)  \Bigr|\!\Bigr|^2_2  \nonumber\\
	&\hspace{-.3cm}+  \sum_{i=1,i\ne k}^{5}  \Bigl|\!\Bigl|  \boldsymbol{x}^a_{d_i} - \big( 1- z^a_i(2) \big) \cdot \boldsymbol{x}_i(2\Delta) - z^a_i(2)   \cdot \boldsymbol{x}^a_i(2\Delta)  \Bigr|\!\Bigr|^2_2  \nonumber\\
	&+  \Bigl|\!\Bigl|  \boldsymbol{x}^a_{d_j} - \boldsymbol{x}_j(\Delta) \Bigr|\!\Bigr|^2_2   +  \Bigl|\!\Bigl|  \boldsymbol{x}^a_{d_k} - \boldsymbol{x}_k(2\Delta) \Bigr|\!\Bigr|^2_2.
	\end{align}
	\normalsize
	
	Now consider the attacker's response to $\boldsymbol{\beta}^d=\{z^d_j(1)=1,z^d_k(2)=1\}$.  Let $\boldsymbol{\beta}^a=\{z^a_m(1)=1,Z^a_n(2)=1\}$ where $m,n \in \mathcal{N}$ is the attacker's response that achieves the minimum cost in (\ref{eq:attackerResponse}). Similar to the defender's actions, in this case $z^a_m(1) =1$ and $z^a_n(1) =1$ and all the other attacker's actions will be zero. Now, rewrite the cost in (\ref{twoStepCostActions2}) with respect to $\boldsymbol{\beta}^a=\{z^a_m(1)=1,z^a_n(2)=1\}$:
	
	\small
	\begin{align} \label{twoStepCostActions3}
	J^a(\boldsymbol{\beta}^d,\boldsymbol{\beta}^a)  & = \sum_{i=1,i\ne m,n}^{5}  \sum_{t=1}^{2} \Bigl|\!\Bigl|  \boldsymbol{x}^a_{d_i} - \boldsymbol{x}_i(t\Delta) \Bigr|\!\Bigr|^2_2   \nonumber\\
	&+  \Bigl|\!\Bigl|  \boldsymbol{x}^a_{d_m} - \boldsymbol{x}^a_m(\Delta) \Bigr|\!\Bigr|^2_2   +  \Bigl|\!\Bigl|  \boldsymbol{x}^a_{d_m} - \boldsymbol{x}_m(2\Delta) \Bigr|\!\Bigr|^2_2  \nonumber \\
	& +  \Bigl|\!\Bigl|  \boldsymbol{x}^a_{d_n} - \boldsymbol{x}_n(\Delta) \Bigr|\!\Bigr|^2_2 +   \Bigl|\!\Bigl|  \boldsymbol{x}^a_{d_n} - \boldsymbol{x}^a_n(2\Delta) \Bigr|\!\Bigr|^2_2.
	\end{align}
	\normalsize
	
	As the cost in (\ref{twoStepCostActions3}) represents the minimum cost in response to $\boldsymbol{\beta}^a=\{z^a_m(1)=1,z^a_n(2)=1\}$, the attacker cannot achieve a better cost by changing its strategy. This minimum cost was achieved by attacking UAV $m$, at the first time step without affecting the other UAVs, and attacking UAV $n$, on the second time step, without affecting the other UAVs. This is because the attacker affects only one UAV, at a time step, and the remaining UAVs travel towards their real destinations. Note that, the attacker's choice at the second time step, i.e, UAV $n$ is independent from its choice at the first time step. After the first time step, UAV $n$ reached its real destination, and yet, this was the best for the attacker at the second time step. Since the action at the second time step is independent from the action at the first time step and it depends only on the new UAVs' locations after the first time step, the attacker will have the same reaction set if faced by the defender's actions sequentially instead of the whole strategy.
	
	This finding can be extended to any number of time steps as the attacker's action, at a time step, will affect only one UAV and its actions in the following time steps will be based on the new UAVs' locations whether they were attacked or not at the previous time step. In other words, when faced by a strategy, the attacker cannot achieve a better outcome than responding at each time step independently. Considering this fact, the defender can determine the reaction set in (\ref{eq:reactionset}) sequentially by solving each time step individually. After determining the complete reaction set, the Stackelberg strategies can be achieved from (\ref{eq:StackelbergSolution}).
	
\end{proof}

\end{appendices}

\vspace{-0.5cm}
\def\baselinestretch{1.00}
\bibliographystyle{IEEEtran}
\vspace{-0.01cm}
\bibliography{references}

\begin{thebibliography}{10}
\providecommand{\url}[1]{#1}
\csname url@samestyle\endcsname
\providecommand{\newblock}{\relax}
\providecommand{\bibinfo}[2]{#2}
\providecommand{\BIBentrySTDinterwordspacing}{\spaceskip=0pt\relax}
\providecommand{\BIBentryALTinterwordstretchfactor}{4}
\providecommand{\BIBentryALTinterwordspacing}{\spaceskip=\fontdimen2\font plus
\BIBentryALTinterwordstretchfactor\fontdimen3\font minus
  \fontdimen4\font\relax}
\providecommand{\BIBforeignlanguage}[2]{{%
\expandafter\ifx\csname l@#1\endcsname\relax
\typeout{** WARNING: IEEEtran.bst: No hyphenation pattern has been}%
\typeout{** loaded for the language `#1'. Using the pattern for}%
\typeout{** the default language instead.}%
\else
\language=\csname l@#1\endcsname
\fi
#2}}
\providecommand{\BIBdecl}{\relax}
\BIBdecl

\bibitem{mozaffari2016unmanned}
M.~Mozaffari, W.~Saad, M.~Bennis, and M.~Debbah, ``Unmanned aerial vehicle with
  underlaid device-to-device communications: Performance and tradeoffs,''
  \emph{IEEE Transactions on Wireless Communications}, vol.~15, no.~6, pp.
  3949--3963, Jun. 2016.

\bibitem{rahmati2019dynamic}
A.~Rahmati, X.~He, I.~Guvenc, and H.~Dai, ``Dynamic mobility-aware interference
  avoidance for aerial base stations in cognitive radio networks,'' \emph{arXiv
  preprint arXiv:1901.02613}, Jan. 2019.

\bibitem{zhang2018downlink}
Z.~Zhang, L.~Li, W.~Liang, X.~Li, A.~Gao, W.~Chen, and Z.~Han, ``Downlink
  interference management in dense drone small cells networks using mean-field
  game theory,'' in \emph{Proceedings of the 10th International Conference on
  Wireless Communications and Signal Processing (WCSP)}, Hangzhou, China, Oct.
  2018, pp. 1--6.

\bibitem{hodgson2016precision}
J.~C. Hodgson, S.~M. Baylis, R.~Mott, A.~Herrod, and R.~H. Clarke, ``Precision
  wildlife monitoring using unmanned aerial vehicles,'' \emph{Scientific
  reports}, vol.~6, p. 22574, Mar. 2016.

\bibitem{eldosouky2017resilient}
A.~Eldosouky, W.~Saad, and N.~Mandayam, ``Resilient critical infrastructure:
  Bayesian network analysis and contract-based optimization,'' \emph{arXiv
  preprint arXiv:1709.00303}, Aug. 2017.

\bibitem{chen2017caching}
M.~Chen, M.~Mozaffari, W.~Saad, C.~Yin, M.~Debbah, and C.~S. Hong, ``Caching in
  the sky: Proactive deployment of cache-enabled unmanned aerial vehicles for
  optimized quality-of-experience,'' \emph{IEEE Journal on Selected Areas in
  Communications (JSAC), Special Issue on Human-In-The-Loop Mobile Networks},
  vol.~35, no.~5, pp. 1046--1061, May 2017.

\bibitem{mozaffari2017mobile}
M.~Mozaffari, W.~Saad, M.~Bennis, and M.~Debbah, ``Mobile unmanned aerial
  vehicles \relax{(UAVs)} for energy-efficient internet of things
  communications,'' \emph{IEEE Transactions on Wireless Communications},
  vol.~16, no.~11, pp. 7574--7589, November 2017.

\bibitem{french2018environment}
A.~French, M.~Mozaffari, A.~Eldosouky, and W.~Saad, ``Environment-aware
  deployment of wireless drones base stations with \relax{G}oogle \relax{E}arth
  simulator,'' in \emph{Proceedings of UNAGI'19 - Workshop on UNmanned aerial
  vehicle Applications in the Smart City}, Kyoto,Japan, Mar 2019.

\bibitem{mozaffari2019beyond}
M.~Mozaffari, A.~T.~Z. Kasgari, W.~Saad, M.~Bennis, and M.~Debbah, ``Beyond
  {5G} with {UAV}s: Foundations of a {3D} wireless cellular network,''
  \emph{IEEE Transactions on Wireless Communications}, vol.~18, no.~1, pp.
  357--372, Jan. 2019.

\bibitem{amer2018caching}
R.~Amer, W.~Saad, H.~ElSawy, M.~Butt, and N.~Marchetti, ``Caching to the sky:
  Performance analysis of cache-assisted comp for cellular-connected {UAV}s,''
  in \emph{Proceedings of the IEEE Wireless Communications and Networking
  Conference (WCNC)}, Marrakech, Morocco, April 2019.

\bibitem{altawy2017security}
R.~Altawy and A.~M. Youssef, ``Security, privacy, and safety aspects of
  civilian drones: A survey,'' \emph{ACM Transactions on Cyber-Physical
  Systems}, vol.~1, no.~2, p.~7, Feb. 2017.

\bibitem{challita2019machine}
U.~Challita, A.~Ferdowsi, M.~Chen, and W.~Saad, ``Machine learning for wireless
  connectivity and security of cellular-connected {UAVs},'' \emph{IEEE Wireless
  Communications}, vol.~26, no.~1, pp. 28--35, Feb. 2019.

\bibitem{sanjab2019game}
A.~Sanjab, W.~Saad, and T.~Ba{\c{s}}ar, ``A game of drones: Cyber-physical
  security of time-critical {UAV} applications with cumulative prospect theory
  perceptions and valuations,'' \emph{arXiv preprint arXiv:1902.03506}, Feb.
  2019.

\bibitem{shepard2012evaluation}
D.~P. Shepard, J.~A. Bhatti, T.~E. Humphreys, and A.~A. Fansler, ``Evaluation
  of smart grid and civilian {UAV} vulnerability to {GPS} spoofing attacks,''
  in \emph{Radionavigation Laboratory Conference Proceedings}, 2012.

\bibitem{kerns2014unmanned}
A.~J. Kerns, D.~P. Shepard, J.~A. Bhatti, and T.~E. Humphreys, ``Unmanned
  aircraft capture and control via {GPS} spoofing,'' \emph{Journal of Field
  Robotics}, vol.~31, no.~4, pp. 617--636, Apr. 2014.

\bibitem{warner2003gps}
J.~S. Warner and R.~G. Johnston, ``Gps spoofing countermeasures,''
  \emph{Homeland Security Journal}, vol.~25, no.~2, pp. 19--27, Dec. 2003.

\bibitem{iqbal2008legal}
M.~U. Iqbal and S.~Lim, ``Legal and ethical implications of {GPS}
  vulnerabilities,'' \emph{J. Int'l Com. L. \& Tech.}, vol.~3, p. 178, 2008.

\bibitem{o2012real}
B.~W. O’Hanlon, M.~L. Psiaki, T.~E. Humphreys, and J.~A. Bhatti, ``Real-time
  spoofing detection using correlation between two civil {GPS} receiver,'' in
  \emph{Proceedings of the ION GNSS Meeting}, Nashville, TN, USA, Sep. 2012.

\bibitem{schmidt2016survey}
D.~Schmidt, K.~Radke, S.~Camtepe, E.~Foo, and M.~Ren, ``A survey and analysis
  of the {GNSS} spoofing threat and countermeasures,'' \emph{ACM Computing
  Surveys (CSUR)}, vol.~48, no.~4, p.~64, May 2016.

\bibitem{jansen2018crowd}
K.~Jansen, M.~Sch{\"a}fer, D.~Moser, V.~Lenders, C.~P{\"o}pper, and J.~Schmitt,
  ``{Crowd-GPS-Sec}: Leveraging crowdsourcing to detect and localize {GPS}
  spoofing attacks,'' in \emph{IEEE Symposium on Security and Privacy (SP)},
  San Francisco, CA, May 2018, pp. 1018--1031.

\bibitem{akos2012s}
D.~M. Akos, ``Who's afraid of the spoofer? {GPS/GNSS} spoofing detection via
  automatic gain control (agc),'' \emph{Navigation: Journal of the Institute of
  Navigation}, vol.~59, no.~4, pp. 281--290, Oct. 2012.

\bibitem{montgomery2011receiver}
P.~Y. Montgomery, ``Receiver-autonomous spoofing detection: Experimental
  results of a multi-antenna receiver defense against a portable civil gps
  spoofer,'' in \emph{Radionavigation Laboratory Conference Proceedings}, 2011.

\bibitem{jansen2016multi}
K.~Jansen, N.~O. Tippenhauer, and C.~P{\"o}pper, ``Multi-receiver gps spoofing
  detection: error models and realization,'' in \emph{Proceedings of the 32nd
  Annual Conference on Computer Security Applications}.\hskip 1em plus 0.5em
  minus 0.4em\relax ACM, 2016, pp. 237--250.

\bibitem{heng2015gps}
L.~Heng, D.~B. Work, and G.~X. Gao, ``Gps signal authentication from
  cooperative peers,'' \emph{IEEE Transactions on Intelligent Transportation
  Systems}, vol.~16, no.~4, pp. 1794--1805, 2015.

\bibitem{ferdowsi2018robust}
A.~Ferdowsi, U.~Challita, W.~Saad, and N.~B. Mandayam, ``Robust deep
  reinforcement learning for security and safety in autonomous vehicle
  systems,'' in \emph{Proceedings of the 21st International Conference on
  Intelligent Transportation Systems (ITSC)}, Maui, Hawaii, USA, Nov. 2018, pp.
  307--312.

\bibitem{qu2011cooperative}
Y.~Qu and Y.~Zhang, ``Cooperative localization against {GPS} signal loss in
  multiple {UAV}s flight,'' \emph{Journal of Systems Engineering and
  Electronics}, vol.~22, no.~1, pp. 103--112, Mar. 2011.

\bibitem{eldosouky2016single}
A.~Eldosouky, W.~Saad, and D.~Niyato, ``Single controller stochastic games for
  optimized moving target defense,'' in \emph{Proceedings of IEEE International
  Conference on Communications (ICC)}, Kuala Lumpur, Malaysia, May 2016, pp.
  1--6.

\bibitem{su2016stealthy}
J.~Su, J.~He, P.~Cheng, and J.~Chen, ``A stealthy {GPS} spoofing strategy for
  manipulating the trajectory of an unmanned aerial vehicle,''
  \emph{IFAC-PapersOnLine}, vol.~49, no.~22, pp. 291--296, Sep. 2016.

\bibitem{ferdowsi2018cyber}
A.~Ferdowsi, S.~Ali, W.~Saad, and N.~B. Mandayam, ``Cyber-physical security and
  safety of autonomous connected vehicles: Optimal control meets multi-armed
  bandit learning,'' \emph{arXiv preprint arXiv:1812.05298}, Dec. 2018.

\bibitem{zeng2017practical}
K.~C. Zeng, Y.~Shu, S.~Liu, Y.~Dou, and Y.~Yang, ``A practical {GPS} location
  spoofing attack in road navigation scenario,'' in \emph{Proceedings of the
  18th International Workshop on Mobile Computing Systems and
  Applications}.\hskip 1em plus 0.5em minus 0.4em\relax Sonoma, CA, USA: ACM,
  Feb. 2017, pp. 85--90.

\bibitem{coxeter1969introduction}
H.~S.~M. Coxeter, H.~S.~M. Coxeter, H.~S.~M. Coxeter, and H.~S.~M. Coxeter,
  \emph{Introduction to geometry}.\hskip 1em plus 0.5em minus 0.4em\relax Wiley
  New York, 1969, vol. 136.

\bibitem{sinha2018review}
A.~Sinha, P.~Malo, and K.~Deb, ``A review on bilevel optimization: from
  classical to evolutionary approaches and applications,'' \emph{IEEE
  Transactions on Evolutionary Computation}, vol.~22, no.~2, pp. 276--295, Apr.
  2018.

\bibitem{han2012game}
Z.~Han, D.~Niyato, W.~Saad, T.~Ba{\c{s}}ar, and A.~Hj{\o}rungnes, \emph{Game
  theory in wireless and communication networks: theory, models, and
  applications}.\hskip 1em plus 0.5em minus 0.4em\relax Cambridge University
  Press, 2012.

\bibitem{simaan1973additional}
M.~Simaan and J.~B. Cruz, ``Additional aspects of the stackelberg strategy in
  nonzero-sum games,'' \emph{Journal of Optimization Theory and Applications},
  vol.~11, no.~6, pp. 613--626, Dec. 1973.

\bibitem{basar1999dynamic}
T.~Ba{\c{s}}ar and G.~J. Olsder, \emph{Dynamic noncooperative game
  theory}.\hskip 1em plus 0.5em minus 0.4em\relax Siam, 1999, vol.~23.

\end{thebibliography}
\vspace{-0.5cm}
\end{document}